\documentclass[submission,copyright,creativecommons]{eptcs}

\usepackage{underscore}           

\usepackage[utf8]{inputenc}
\usepackage[english]{babel}
\usepackage{amsmath}
\usepackage{amsfonts}
\usepackage{amssymb}
\usepackage{amsthm}
\usepackage{graphicx}

\usepackage{hyperref}

\usepackage[pdftex,dvipsnames]{xcolor}

\usepackage{tikz}
\usetikzlibrary{shapes}
\usetikzlibrary{automata}

\usetikzlibrary{shapes.multipart}
\usetikzlibrary{decorations.pathmorphing}

\usepackage{caption}

\def\rest#1#2{#1_{\restriction#2}}
\def\outcome#1#2{\langle #1 \rangle_{#2}}

\DeclareMathOperator{\Succ}{Succ}
\DeclareMathOperator{\Min}{min}
\DeclareMathOperator{\Max}{max}
\DeclareMathOperator{\Occ}{Occ}
\DeclareMathOperator{\Inf}{Inf}

\DeclareMathOperator{\Hist}{Hist}
\DeclareMathOperator{\Plays}{Plays}

\DeclareMathOperator{\First}{First}

\DeclareMathOperator{\Gain}{Gain}
\DeclareMathOperator{\Win}{Obj}

\newcommand{\fixpointalgo}{fixpoint algorithm}
\newcommand{\init}{init}
\newcommand{\W}{\Omega}

\DeclareMathOperator{\rP}{\mathbf{P}}
\DeclareMathOperator{\Wit}{\mathcal{P}}

\usepackage[colorinlistoftodos,prependcaption,textsize=tiny]{todonotes}
\usepackage{xargs}                      

\newcommandx{\unsure}[2][1=]{\todo[linecolor=red,backgroundcolor=red!25,bordercolor=red,#1]{#2}}
\newcommandx{\change}[2][1=]{\todo[linecolor=blue,backgroundcolor=blue!25,bordercolor=blue,#1]{#2}}
\newcommandx{\info}[2][1=]{\todo[linecolor=BlueGreen,backgroundcolor=OliveGreen!25,bordercolor=OliveGreen,#1]{#2}}
\newcommandx{\improvement}[2][1=]{\todo[linecolor=SpringGreen,backgroundcolor=SpringGreen!25,bordercolor=SpringGreen,#1]{#2}}
\newcommandx{\thiswillnotshow}[2][1=]{\todo[disable,#1]{#2}}

\theoremstyle{definition}
\newtheorem{defi}{Definition}
\newtheorem{example}{Example}

\theoremstyle{theorem}
\newtheorem{thm}{Theorem}
\newtheorem{prop}{Proposition}
\newtheorem{lemma}{Lemma}
\newtheorem{corollary}{Corollary}

\theoremstyle{remark}
\newtheorem{remark}{Remark}

\makeatletter
\let\c@defi\c@thm
\let\c@lemma\c@thm
\let\c@prop\c@thm

\let\c@corollary\c@thm
\let\c@remark\c@thm
\let\c@example\c@thm
\let\c@proposition\c@thm
\makeatother

\author{Thomas Brihaye
\institute{Université de Mons (UMONS), Belgium}
\email{thomas.brihaye@umons.ac.be}
\and
Véronique Bruyère
\institute{Université de Mons (UMONS), Belgium}
\email{veronique.bruyere@umons.ac.be}
\and
Aline Goeminne
\institute{Université de Mons (UMONS), Belgium \\
Université libre de Bruxelles (ULB), Belgium }
\email{aline.goeminne@umons.ac.be}
\and
Jean-François Raskin\thanks{Author supported by ERC Starting Grant (279499: inVEST).}
\institute{Université libre de Bruxelles (ULB), Belgium}
\email{jraskin@ulb.ac.be}
}

\title{Constrained existence problem for weak subgame perfect equilibria with $\omega$-regular Boolean objectives (full version)\thanks{This article is based upon work from COST Action GAMENET CA 16228 supported by COST (European Cooperation in Science and Technology).}}
\begin{document}
\maketitle

\begin{abstract}
We study multiplayer turn-based games played on a finite directed graph such that each player aims at satisfying an $\omega$-regular Boolean objective. Instead of the well-known notions of Nash equilibrium (NE) and subgame perfect equilibrium (SPE), we focus on the recent notion of weak subgame perfect equilibrium (weak SPE), a refinement of SPE. In this setting, players who deviate can only use the subclass of strategies that differ from the original one on a finite number of histories. We are interested in the constrained existence problem for weak SPEs. We provide a complete characterization of the computational complexity of this problem: it is P-complete for Explicit Muller objectives, NP-complete for Co-B\"uchi, Parity, Muller, Rabin, and Streett objectives, and PSPACE-complete for Reachability and Safety objectives (we only prove NP-membership for B\"uchi objectives). We also show that the constrained existence problem is fixed parameter tractable and is polynomial when the number of players is fixed. All these results are based on a fine analysis of a fixpoint algorithm that computes the set of possible payoff profiles underlying weak SPEs.
\end{abstract}

\section{Introduction}

\emph{Two-player zero-sum graph games} with $\omega$-regular objectives are the classical mathematical model to formalize the reactive synthesis problem~\cite{PnueliR89,Thomas95}. More recently, generalization from zero-sum to non zero-sum, and from two players to $n$ players have been considered in the literature, see e.g.~\cite{berwanger07,BrenguierCHPRRS16,BRS-concur15,BRS14,BMR14,KHJ06,FismanKL10,KupfermanPV14,Ummels06} and the surveys~\cite{Bruyere17,GU08}. 
Those extensions are motivated by two main limitations of the classical setting. First, zero-sum games assume a fully antagonistic environment while this is often not the case in practice: the environment usually has its own goal.  While the fully antagonistic assumption is simple and sound (a winning strategy against an antagonistic environment is winning against any environment that pursues its own objective), it may fail to find a winning
strategy even if solutions exist when the objective of the environment is accounted.
Second, modern reactive systems are often composed of several modules, and each module has its own specification and should be considered as a player on its own right. This is why we need to consider \emph{$n$-player graph games}.

For $n$-player graph games, solution concepts like \emph{Nash equilibria} (NEs)~\cite{nash50} are natural to consider. 
A strategy profile is an NE if no player has an incentive to deviate unilaterally from his strategy, \emph{i.e.} no player can strictly improve on the outcome of the strategy profile by changing his strategy only.
In the context of sequential games (such as games played on graphs), NEs allow for non-credible threats that rational players should not carry out. To avoid non-credible threats, refinements such as \emph{subgame perfect equilibria} (SPEs)~\cite{osbornebook} have been advocated. A strategy profile is an SPE if it is an NE in all the subgames of the original game. So players need to play rationally in all subgames, and this ensures that non-credible threats cannot exist. For applications of this concept to $n$-player graph games, we refer the reader to~\cite{BrihayeBDG12,KHJ06,Ummels06}. 

In~\cite{BrihayeBMR15}, the notion of {\em weak} subgame perfect equilibrium (weak SPE) is introduced, and it is shown how it can be used to study the existence SPEs (possibly with contraints) in quantitative reachability games.  While an SPE must be resistant to any unilateral deviation of one player, a weak SPE must be resistant to deviations restricted to deviating strategies that differ from the original one on a \emph{finite number} of histories only. In~\cite{Bruyere0PR17} the authors study general conditions on the structure of the game graph and on the preference relations of the players that guarantee the existence of a weak SPE for quantitative games. Weak SPEs retain most of the important properties of SPEs and they coincide with them when the payoff function of each player is continuous (see e.g.~\cite{fudenberg1991game}). Weak SPEs are also easier to characterize and to manipulate algorithmically. We refer the interested reader to~\cite{BrihayeBMR15,Bruyere0PR17} for further justifications of their interest, as well as for related work on NEs and SPEs.

\paragraph{Main contributions}
In this paper, we concentrate on graph games with \emph{$\omega$-regular Boolean objectives}. While SPEs, and thus weak SPEs, are always guaranteed to exist in such games, we here study the computational complexity of the {\em constrained existence problem} for weak SPEs, \emph{i.e.} equilibria in which some designated players have to win and some other ones have to loose. More precisely, our main results are as follows:
\begin{itemize}
  	\item We study the constrained existence problem for games with Reachability, Safety, B\"uchi, Co-B\"uchi, Parity, Explicit Muller, Muller, Rabin, and Streett objectives. We provide a \emph{complete characterization} of the computational complexity of this problem for all the classes of objectives with one exception: B\"uchi objectives. The problem is P-complete for Explicit Muller objectives, it is NP-complete for Co-B\"uchi, Parity, Muller, Rabin, and Streett objectives, and it is PSPACE-complete for Reachability and Safety objectives. In case of B\"uchi objectives, we show membership to NP but we fail to prove hardness. 
	\item Our complexity results rely on the identification of a \emph{symbolic witness} for the constrained existence of a weak SPE, the size of which allows us to prove NP/PSPACE-membership. As the constrained existence problem is PSPACE-complete for Reachability and Safety objectives, symbolic witnesses as compact as those for the other objectives cannot exist unless NP $=$ PSPACE. The identification of symbolic witnesses is obtained thanks to a \emph{fixpoint algorithm} that computes the set of all possible payoff profiles underlying weak SPEs.
	\item When the number of players is fixed, we show that the constrained existence problem can be solved in \emph{polynomial} time for all $\omega$-regular objectives. We also prove that it is \emph{fixed parameter tractable} where the parameter is the number of players, for Reachability, Safety, B\"uchi, Co-B\"uchi, and Parity objectives. For Rabin, Streett, and Muller objectives, we still establish fixed parameter tractability but we need to consider some additional parameters depending on the objectives. These tractability results are obtained by a fine analysis of the complexity of the fixpoint algorithm mentioned previously. 
\end{itemize}

\paragraph{Related work and additional contributions} 
In~\cite{GU08,Ummels06}, a tree automata-based algorithm is given to decide the constrained existence problem for SPEs on graph games with $\omega$-regular objectives defined by parity conditions. A complexity gap is left open: this algorithm executes in EXPTIME and NP-hardness of the decision problem is proved. In this paper, we focus on weak SPEs for which we provide precise complexity results for the constrained existence problem. We also observe that our results on Reachabilty and Safety objectives transfer from weak SPEs to SPEs: the constrained existence problem for SPEs is PSPACE-complete for those objectives. {\em Quantitative} Reachability objectives are investigated in~\cite{BrihayeBMR15} where it is proved that the constrained existence problem for weak SPEs and SPEs is decidable, but its exact complexity is left open.

In~\cite{BrihayeBMR15,Bruyere0PR17,FleschKMSSV10}, the existence of (weak) SPEs in graph games is established using a construction based on a fixpoint.  Our fixpoint algorithm is mainly inspired by the fixpoint technique of~\cite{Bruyere0PR17}. However, we provide complexity results based on this fixpoint while transfinite induction is used in~\cite{Bruyere0PR17}. Furthermore, we have modified the technique of~\cite{Bruyere0PR17} in a way to get a fixpoint that contains exactly all the possible payoff profiles of weak SPEs. This is necessary to get a decision algorithm for the constrained existence problem. 
 
Profiles of strategies with finite-memory are more appealing from a practical point of view. It is shown in~\cite{Ummels06} that when there exists an SPE in a graph game with $\omega$-regular objectives, then there exists one that uses finite-memory strategies and has the same payoff profile. Thanks to the symbolic witnesses, we have refined those results for weak SPEs.

\paragraph{Structure of the paper}  In Section~\ref{sec:prelim}, we recall the notions of $n$-player graph games and of (weak) SPE, and we state the studied constrained existence problem. In Section~\ref{section:charac}, we provide a fixpoint algorithm that computes all the possible payoff profiles for weak SPEs on a given graph game. From this fixpoint, we derive symbolic witnesses of weak SPEs. In Section~\ref{sec:classes}, we study the complexity classes of the constrained existence problem for all objectives except Explicit Muller objectives. In Section~\ref{sec:FPT}, we prove the fixed parameter tractability of the constrained existence problem and we show that is in polynomial time when the number of players is fixed. We also show that this problem it is P-complete for Explicit Muller objectives. In Section~\ref{sec:conc}, we give a conclusion and propose future work.

\section{Preliminaries}
\label{sec:prelim}

In this section, we introduce multiplayer graph games in which each player aims to achieve his Boolean objective. We focus on classical $\omega$-regular objectives, like Reachability, B\"uchi, aso. We recall two classical concepts of equilibria: Nash equilibrium and subgame perfect equilibrium~(see \cite{GU08}). We also recall weak variants of these equilibria as proposed in~\cite{BrihayeBMR15,Bruyere0PR17}. We conclude the section by the constrained existence problem that is studied in this paper.

\subsection{Multiplayer Boolean games}

\begin{defi}[Boolean game]
A \emph{multiplayer Boolean game} is a tuple $\mathcal{G} = (\Pi, V, (V_i)_{i \in \Pi}, E, (\Gain_i)_{i \in \Pi})$ where 
\begin{itemize}
\item $\Pi = \{ 1,2, \ldots,n \}$ is a finite set of $n$ \emph{players};
\item $G = (V,E)$ is a finite directed graph and for all $v \in V$ there exists $v' \in V$ such that $(v,v') \in E$;
\item $(V_i)_{i \in \Pi}$ is a partition of $V$ between the players;
\item $\Gain = (\Gain_i)_{i \in \Pi}$ is a tuple of functions $\Gain_i: V ^\omega \rightarrow  \{0,1\}$ that assigns a Boolean value  to each infinite path of $G$ for player $i$.
\end{itemize}
\end{defi}

A \emph{play} in $\mathcal{G}$ is an infinite  sequence of vertices $\rho = \rho_0 \rho_1 \ldots$ such that for all $k \in \mathbb{N}$, $(\rho_k, \rho_{k+1}) \in E$.  A \emph{history} is a finite sequence $h = h_0h_1 \ldots h_n$ ($n \in\mathbb{N}$) defined similarly. We denote the set of plays by $\Plays$ and the set of histories by $\Hist$. Moreover, the set $\Hist_i$ is the set of histories such that the last vertex $v$ is a vertex of player $i$, i.e. $v \in V_i$. The \emph{length} $|h|$ of $h$ is the number $n$ of its edges. A play $\rho$ is called a \emph{lasso} if it is of the form $\rho = h\ell^\omega$ with $h\ell \in \Hist$. Notice that $\ell$ is not necessary a simple cycle. The \emph{length of a lasso} $h\ell^\omega$ is the length of $h\ell$. For all $h\in \Hist$, we denote by $\First(h)$ the first vertex $h_0$ of $h$. We use notation $h < \rho$ when a history $h$ is prefix of a play (or a history)  $\rho$. Given a play $\rho = \rho_0\rho_1 \ldots$, the set $\Occ(\rho) = \{ v \in V \mid \exists k, \rho_k = v \}$ is the set of vertices \emph{visited} by $\rho$, and $\Inf(\rho) = \{ v \in V \mid \forall k, \exists j \geq k, \rho_j = v \}$ is the set of vertices \emph{infinitely often visited} by $\rho$. Given a vertex $v \in V$, $\Succ(v) = \{v' \mid (v, v') \in E \}$ is the set of successors of $v$, and $\Succ^*(v)$ is the set of vertices reachable from $v$ in $G$.

When an \emph{initial} vertex $v_0\in V$ is fixed, we call $(\mathcal{G}, v_0)$ an \emph{initialized game}. A play (resp. a history) of $(\mathcal{G},v_0)$ is a play (resp. a history) of $\mathcal{G}$ starting in $v_0$. The set of such plays (resp. histories) is denoted by $\Plays(v_0)$ (resp. $\Hist(v_0)$). We also use notation $\Hist_i(v_0)$ when these histories end in a vertex $v \in V_i$.

The goal of each player $i$ is to achieve his objective, \emph{i.e.}, to maximize his gain.

\begin{defi}[Objective]
	\label{defi:winningCondition}
	For each player $i \in \Pi$, let $\Win_i \subseteq V^\omega$ be his \emph{objective}. In the setting of multiplayer Boolean game, the gain function $\Gain_i$ is defined such that $\Gain_i(\rho) = 1$ (resp. $\Gain_i(\rho) = 0$) if and only if $\rho \in \Win_i$ (resp. $\rho \not \in \Win_i$).
	\end{defi}
	
An objective $\Win_i$ (or the related gain function $\Gain_i$) is \emph{prefix-independent} if for all $h \in V^*$ and $ \rho \in V^{\omega}$, we have $\rho \in \Win_i$ if and only if $h\rho \in \Win_i$. In this paper, we focus on classical \emph{$\omega$-regular} objectives: Reachability,  Safety, Büchi, Co-Büchi, Parity, Explicit Muller, Muller, Rabin, and Streett and we suppose that each player has the \emph{same type} of objective. For instance, we say that $\mathcal{G}$ is a \emph{Boolean game with B\"uchi objectives} to express that all players have a B\"uchi objective. 

\begin{defi}[Classical $\omega$-regular objective]
The set $\Win_i$ is a \emph{Reachability, Safety, Büchi, Co-Büchi, Parity, Explicit Muller, Muller, Rabin}, or \emph{Streett} objective for player~$i$ if and only if $\Win_i$ is composed of the plays $\rho$ satisfying:
\begin{itemize}
	\item \emph{Reachability}:  given $F \subseteq V$, $\Occ(\rho) \cap F \neq \emptyset$;
	\item \emph{Safety}: given $F \subseteq V$, $ \Occ(\rho) \cap F = \emptyset$;
	\item \emph{Büchi}: given $F \subseteq V$, $ \Inf(\rho) \cap F \neq \emptyset $;
	\item \emph{Co-Büchi}: given $F \subseteq V$, $ \Inf(\rho) \cap F = \emptyset$;
	\item \emph{Parity}: given a coloring function $\Omega : V \rightarrow \{1,\ldots,d\}$, $ \Max(\Inf(\Omega(\rho)))\footnote{Where $\Omega(\rho) = \Omega(\rho_0)\Omega(\rho_1)\ldots \Omega(\rho_n)\ldots$.} \text{ is even}$;
	\item \emph{Explicit Muller}: given $\mathcal{F}  \subseteq 2^V$, $\Inf(\rho)\in \mathcal{F}$; 
	\item \emph{Muller}: given a coloring function $\Omega : V \rightarrow \{1,\ldots,d\}$, and $\mathcal{F} \subseteq 2^{\Omega(V)}$, $\Inf(\Omega(\rho)) \in \mathcal{F}$;
	\item \emph{Rabin}: given $(G_j,R_j)_{1\leq j \leq k}$ a family of pair of sets $G_j,R_j \subseteq V$,\\there exists $j \in {1,\ldots,k}$ such that $\Inf(\rho) \cap G_j \neq \emptyset$ and $\Inf(\rho) \cap R_j = \emptyset$;
	\item \emph{Streett}: given $(G_j,R_j)_{1\leq j \leq k}$ a family of pair of sets $G_j,R_j \subseteq V$,\\ for all $j \in {1,\ldots,k}$,  $\Inf(\rho) \cap G_j = \emptyset$ or $\Inf(\rho) \cap R_j \neq \emptyset$.
\end{itemize}
\end{defi}

\noindent
All these objectives are prefix-independent except Reachability and Safety objectives.

A \emph{strategy} of a player $i\in \Pi$ is a function $\sigma_i: \Hist_i \rightarrow V$. This function assigns to each history $hv$ with $v \in V_i$, a vertex $v'$ such that $(v,v') \in E$. In an initialized game $(\mathcal{G},v_0)$, $\sigma_i$ needs only to be defined for histories starting in $v_0$. A play $\rho=\rho_0\rho_1\ldots$ is \emph{consistent} with  $\sigma_i$ if for all $\rho_k \in V_i$ we have that $\sigma_i(\rho_0 \ldots \rho_k) = \rho_{k+1}$. A strategy $\sigma_i$ is \emph{positional} if it only depends on the last vertex of the history, \emph{i.e.}, $\sigma_i(hv) = \sigma_i(v)$ for all $hv \in \Hist_i$. It is \emph{finite-memory} if it can be encoded by a deterministic \emph{Moore machine} ${\cal M} = (M, m_0, \alpha_u, \alpha_n)$ where $M$ is a finite set of states (the memory of the strategy), $m_0 \in M$ is the initial memory state, $\alpha_u\colon M \times V \rightarrow M$ is the update function, and $\alpha_n\colon M \times V_i \rightarrow V$ is the next-action function. The Moore machine $\cal M$ defines a strategy $\sigma_i$ such that $\sigma_i(h v) = \alpha_n(\widehat{\alpha}_u(m_0,h),v)$ for all histories $h v \in \Hist_i$, where $\widehat{\alpha}_u$ extends $\alpha_u$ to histories as expected. The \emph{size} of the strategy $\sigma_i$ is the size $|M|$ of its machine $\cal M$. Note that $\sigma_i$ is positional  when $|M| = 1$.

A \emph{strategy profile} is a tuple $\sigma = (\sigma_i)_{i\in \Pi}$ of strategies, one for each player. It is called positional (resp. finite-memory) if for all $i \in \Pi$, $\sigma_i$ is positional (resp. finite-memory).  Given an initialized game $(\mathcal{G}, v_0)$ and a strategy profile $\sigma$, there exists an unique play from $v_0$ consistent with each strategy $\sigma_i$. We call this play the \emph{outcome} of $\sigma$ and it is denoted by $\outcome{\sigma}{v_0}$. Let $p = (p_i)_{i \in \Pi} \in \{0,1\}^\Pi$, we say that $\sigma$ is a strategy profile \emph{with payoff} $p$ or that $\outcome{\sigma}{v_0}$ \emph{has payoff} $p$ if $p_i = \Gain_i(\outcome{\sigma}{v_0})$ for all $i \in \Pi$.

\subsection{Solution concepts}

In the multiplayer game setting, the solution concepts usually studied are \emph{equilibria} (see~\cite{GU08}). We here recall the concepts of Nash equilibrium and subgame perfect equilibrium, as well as some variants. We begin with the notion of deviating strategy.

Let $\sigma = (\sigma_i)_{i\in \Pi}$ be a strategy profile in an initialized Boolean game $(\mathcal{G},v_0)$. Given $i \in \Pi$, a strategy $\sigma'_i \neq \sigma_i$ is a \emph{deviating} strategy of player~$i$, and $(\sigma'_i, \sigma_{-i})$ denotes the strategy profile $\sigma$ where $\sigma'_i$ replaces $\sigma_i$. Such a strategy is a \emph{profitable deviation} for player~$i$ if $\Gain_i(\outcome{\sigma}{v_0}) < \Gain_i(\outcome{\sigma'_i, \sigma_{-i}}{v_0})$. We say that $\sigma'_i$ is  \emph{finitely deviating} from $\sigma_i$ if  $\sigma'_i$ and $\sigma_i$ only differ on a finite number of histories, and that $\sigma'_i$ is \emph{one-shot deviating} from $\sigma_i$ if $\sigma'_i$ and $\sigma_i$ only differ on $v_0$~\cite{BrihayeBMR15,Bruyere0PR17}.  

The notion of Nash equilibrium (NE) is classical: a strategy profile $\sigma$ in an initialized game $(\mathcal{G},v_0)$ is a \emph{Nash equilibrium} if no player has an incentive to deviate unilaterally from his strategy since he has no profitable deviation, \emph{i.e.}, for each $i \in \Pi$ and each deviating strategy $\sigma'_i$ of player $i$ from $\sigma_i$, the following inequality holds: $\Gain_i(\outcome{\sigma}{v_0}) \geq \Gain_i(\outcome{\sigma'_i, \sigma_{-i}}{v_0})$. In this paper we focus on two variants of NE:  weak/very weak NE~\cite{BrihayeBMR15,Bruyere0PR17}. 

\begin{defi}[Weak/very weak Nash equilibrium]
	A strategy profile $\sigma$ is a \emph{weak NE} (resp. \emph{very weak NE}) in $(\mathcal{G},v_0)$ if, for each player $i\in \Pi$, for each finitely deviating (resp. one-shot) strategy $\sigma'_i$ of player~$i$ from $\sigma_i$, we have $\Gain_i(\outcome{\sigma}{v_0}) \geq \Gain_i(\outcome{\sigma'_i, \sigma_{-i}}{v_0})$.
\end{defi}

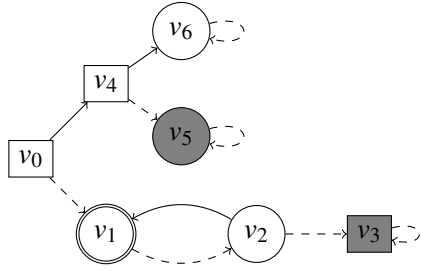
\begin{figure}
	\centering
	\scalebox{1}{

	\begin{tikzpicture}[]
		\node[draw] (v0) at (0,0){$v_0$};
		\node[draw,accepting,circle] (v1) at (1,-1){$v_1$};
		\node[draw,circle] (v2) at (3,-1){$v_2$};
		\node[draw, fill=gray] (v3) at (4.5,-1){$v_3$};
		\node[draw] (v4) at (1,1){$v_4$};
		\node[draw,circle, fill=gray] (v5) at (2,0.3){$v_5$};
		\node[draw,circle] (v6) at (2,1.7){$v_6$};

		\draw[->,dashed] (v0) to (v1);
		\draw[->,dashed] (v1) to [bend right] (v2);
		\draw[->] (v2) to [bend right] (v1);

		\draw[->,dashed] (v2) to (v3);
		\draw[->] (v0) to  (v4);

		\draw[->,dashed] (v3) to [loop right] (v3);
		\draw[->,dashed] (v4) to (v5);
		\draw[->] (v4) to (v6);
		\draw[->,dashed] (v5) to [loop right] (v5);
		\draw[->,dashed] (v6) to [loop right] (v6);

	\end{tikzpicture}}
	\caption{Example of a Boolean game with B\"uchi objectives}
	\label{fig:weakNEvsNE}   
%
%
%
%
%
\end{figure}

\begin{example}
\label{ex:weakNE}
Figure~\ref{fig:weakNEvsNE} illustrates an initialized Boolean game $(\mathcal{G},v_0)$ with B\"uchi objectives  in which there exists a weak NE that is not an NE. In this game, player~1 (resp. player~2) owns round (resp. square) vertices and wants to visits $v_1$ (resp. $v_3$ or $v_5$) infinitely often. The positional strategy profile $\sigma = (\sigma_1,\sigma_2)$ is depicted by dashed arrows, its outcome is equal to $\outcome{\sigma}{v_0}=v_0v_1v_2v_3^\omega$, and $\sigma$ has payoff $(0,1)$. Notice that player 1 has an incentive to deviate from his strategy $\sigma_1$ with a strategy $\sigma'_1$ that goes to $v_1$ for all histories ending in $v_2$. This is indeed a profitable deviation for him since $\Gain(\outcome{(\sigma'_1, \sigma_2)}{v_0}) = (1,0)$. So, $\sigma$ is not an NE. Nevertheless, it is a weak NE because $\sigma'_1$ is the only profitable deviation and it is not finitely deviating (it differs from $\sigma_1$ on all histories of the form $v_0(v_1v_2)^k$ for $k \in \mathbb{N}\backslash \{0 \}$).
\end{example} 

When considering games played on graphs, a well-known refinement of NE is the concept of \emph{subgame perfect equilibrium} (SPE) which a strategy profile being an NE in each subgame. Variants of weak/very weak SPE can also be studied as done with NEs. Formally, given an initialized Boolean game $(\mathcal{G},v_0)$ and a history $hv \in \Hist(v_0)$, the initialized game $(\rest{\mathcal{G}}{h},v)$ is called a \emph{subgame}\footnote{Notice that $(\mathcal{G},v_0)$ is subgame of itself.} of $(\mathcal{G},v_0)$ such that $\rest{\mathcal{G}}{h} = (\Pi, V, (V_i)_{i\in \Pi}, E, \rest{\Gain}{h})$ and $\Gain_{i\restriction h}(\rho) = \Gain_i(h\rho)$ for all $i \in \Pi$ and $\rho \in V^{\omega}$. Moreover if $\sigma_i$ is a strategy for player~$i$ in $(\mathcal{G},v_0)$, then $\sigma_{i\restriction h}$ denotes the strategy in $(\rest{\mathcal{G}}{h},v)$ such that for all histories $h'\in \Hist_i(v)$, $\sigma_{i\restriction h}(h') = \sigma_i(hh')$. Similarly, from a strategy profile $\sigma$ in $(\mathcal{G},v_0)$, we derive the strategy profile $\rest{\sigma}{h}$ in $(\rest{\mathcal{G}}{h},v)$. The play $\outcome{\rest{\sigma}{h}}{v}$ is called a \emph{subgame outcome} of $\sigma$. 

\begin{defi}[Subgame perfect equilibrium and weak/very weak subgame perfect equilibrium]
	A strategy profile $\sigma$ is a \emph{(resp. weak, very weak) subgame perfect equilibrium} in $(\mathcal{G},v_0)$ if for all $hv \in \Hist(v_0)$, $\rest{\sigma}{h}$ is a (resp. weak, very weak) NE in $(\rest{\mathcal{G}}{h},v)$.
\end{defi}

When one needs to show that a strategy profile is a weak SPE, the next proposition is very useful because it states that it is enough to consider one-shot deviating strategies.  
\begin{prop}[\cite{BrihayeBMR15}] \label{prop:equivWSPEetVWSPE}
	A strategy profile $\sigma$ is a weak SPE if and only if $\sigma$ is a very weak SPE.
\end{prop}

\begin{example}
\label{ex:weakSPE}
In Example~\ref{ex:weakNE} is given a weak NE $\sigma$ in the game $(\mathcal{G},v_0)$ depicted in Figure~\ref{fig:weakNEvsNE}. This strategy profile is also a very weak SPE (and thus a weak SPE by Proposition~\ref{prop:equivWSPEetVWSPE}). For instance, in the subgame $(\rest{\mathcal{G}}{h},v)$ with $h = v_0v_1$ and $v = v_2$, the only one-shot deviating strategy $\sigma'_1$ is such that $\sigma'_1$ coincides with $\sigma_{1\restriction h}$ except that $\sigma'_1(v_2) = v_1$. This is not a profitable deviation for player~$1$ in $(\rest{\mathcal{G}}{h},v)$. Notice that $\sigma$ is not an SPE since it is not an NE as explained in Example~\ref{ex:weakNE}.

\end{example}

In general, the notions of SPE and weak SPE are not equivalent (see Example~\ref{ex:weakSPE}). Nevertheless they coincide for the class of Boolean games with Reachability objectives. 
%

\begin{prop} 
\label{prop:reach}
Let $\sigma$ be a strategy profile in an initialized Boolean game $(\mathcal{G},v_0)$ with Reachability objectives. Then $\sigma$ is an SPE if and only if $\sigma$ is a weak SPE.
\end{prop}

\begin{proof} Each player $i$ has a Reachability objective, let $F_i$ be the set of vertices he aims to visit. 
	
		$(\Rightarrow)$ This implication is a consequence of the definitions of SPE and weak SPE.
		
		$(\Leftarrow)$ Let $\sigma$ be a  weak SPE in $(\mathcal{G},v_0)$. Assume that $\sigma$ is not an SPE, \emph{i.e.}, there exists $hv \in \Hist(v_0)$ such that $\rest{\sigma}{h}$ is not an NE in $(\rest{\mathcal{G}}{h},v)$. Then some player~$i$ has a profitable deviation $\sigma_i'$ in the subgame $(\rest{\mathcal{G}}{h},v)$. As $\Gain_i$ takes its values in $\{0,1\}$, this means that 
		$$0 = \Gain_i(h \rho) < \Gain_i(h \rho') = 1$$
with $\rho = \outcome{\rest{\sigma}{h}}{v}$ and $\rho' = \outcome{\sigma_i', \sigma_{-i\restriction h}}{v}$. We consider the first occurrence of a vertex of $F_i$ along $h\rho'$ (which appears in $\rho'$ and not in $h$ as $\Gain_i(h \rho) = 0$): let $g'$ of mininal length such that $hg' < h\rho'$ and $g'$ ends in some $v' \in F_i$.
Let us define a strategy $\tau_i$ that is finitely deviating from $\sigma_{i\restriction h}$ and profitable for player~$i$ in $(\rest{\mathcal{G}}{h},v)$. This will be in contradiction with our hypothesis. For all $g \in \Hist_i(v)$, let
$$\tau_i(g) = \begin{cases} \sigma'_i(g) & \text{ if } g \leq g' \\
\sigma_{i\restriction h}(g) & \text{ otherwise } \end{cases}.$$
By definition of $\tau_i$, we have that $\Gain_i(h \outcome{\tau_i,\sigma_{-i\restriction h}}{v}) = \Gain_i(h \rho') = 1$ and $\tau_i$ is finitely deviating from $\sigma_{i\restriction h}$ since $|g'|$ is finite. 
\end{proof}

\subsection{Constraint problem}

It is proved in~\cite{Bruyere0PR17} that there always exists a weak SPE in Boolean games.  In this paper, we are interested in solving the following \emph{constraint problem}:

\begin{defi}[Constraint problem]
\label{decidabilityProblem}
Given $(\mathcal{G},v_0)$ an initialized Boolean game and thresholds $x,y \in \{0,1\}^{|\Pi|}$, decide whether there exists a weak SPE in $(\mathcal{G},v_0)$ with payoff $p$  such that $x \leq p \leq y$.\footnote{The order $\leq$ is the componentwise order, that is, $x_i \leq p_i \leq y_i$, for all $i \in \Pi$.}
\end{defi}



In the next sections, we solve the constraint problem for the classical $\omega$-regular objectives. The complexity classes that we obtain are shown in Table~\ref{tab:complexity}. They are detailed in Section~\ref{sec:classes} with the case of Explicit Muller objectives postponed to Section~\ref{sec:ExplicitMuller}. The case of B\"uchi objectives remains open, since we only propose a non-deterministic algorithm in polynomial time but no matching lower bound. In Section~\ref{sec:FPT}, we prove that the constraint problem for weak SPEs is fixed parameter tractable and becomes polynomial when the number of players is fixed. All these results are based on a characterization of the set of possible payoffs of a weak SPE, that is described in Section~\ref{section:charac}. 

\begin{table}[h!]
\centering
\caption{Complexity classes of the constraint problem for $\omega$-regular objectives}
\scalebox{0.9}{
\begin{tabular}{|l|c|c|c|c|c|c|c|c|}
\hline
                & Explicit Muller & Co-Büchi  & Parity  & Muller & Rabin & Streett & Reachability & Safety  \\ \hline\hline
P-complete       & $\times$  &   &   &       &   &   &   &        \\ \hline
NP-complete    &   		& $\times$ & $\times$ &  $\times$  & $\times$      &  $\times$     &     &  \\ \hline
PSPACE-complete      &   &  &  &        &   &    & $\times$  & $\times$\\ \hline
\end{tabular}}
\label{tab:complexity}
\end{table}

\section{Characterization}
\label{section:charac}

In this section our aim is twofold: first, we characterize the set of possible payoffs of weak SPEs and second, we show how it is possible to build a weak SPE given a set of lassoes with some ``good properties". Those characterizations work for Boolean games with \emph{prefix-independent gain functions}. We make this hypothesis all along Section~\ref{section:charac}.

\subsection{Remove-Adjust procedure}
\label{sec:RemoveAdjust}

Let $(\mathcal{G},v_0)$ be an initialized Boolean game with prefix-independent gain functions. The computation of the set of all the payoffs of weak SPEs in $(\mathcal{G},v_0)$ is inspired by a fixpoint procedure explained in~\cite{Bruyere0PR17}. Each vertex $v$ is \emph{labeled} by a set of payoffs $p \in \{0,1\}^{|\Pi|}$. 
Initially, these payoffs are those for which there exists a play in $\Plays(v)$ with payoff $p$. Then step by step, some payoffs are removed for the labeling of $v$ as soon as we are sure they cannot be the payoff of $\rest{\sigma}{h}$ in a subgame $(\rest{\mathcal{G}}{h},v)$ for some weak SPE $\sigma$.\footnote{The value of $h$ is not important since the gain functions are prefix independent. This is why we only focus on $v$ and not on~$hv$.} 
When a fixpoint is reached, the labeling of the initial vertex $v_0$ exactly contains all the payoffs of weak SPEs in $(\mathcal{G},v_0)$. Hence, at each step of this procedure, the payoffs labeling a vertex $v$ are payoffs of \emph{potential} subgame outcomes of a weak SPE. Their number decreases until reaching a fixpoint. 

We formally proceed as follows. For all $v \in V$, we define the initial labeling of $v$ as:
$$\rP_0(v) = \{ p \in \{0,1\}^{|\Pi|} \mid \text{ there exists } \rho \in \Plays(v) \text{ such that } \Gain(\rho) = p \}.$$

Then for each step $k \in \mathbb{N} \setminus \{0\}$, we compute the set $\rP_k(v)$ by alternating between two operations: \emph{Remove} and \emph{Adjust}. To this end, we need to introduce the notion of $(p,k)$-labeled play. Let $p$ be a payoff and $k$ be a step, a play $\rho= \rho_0\rho_1\rho_2 \ldots$ is \emph{$(p,k)$-labeled} if for all $j \in \mathbb{N}$ we have $p \in \rP_k(\rho_j)$, that is, $\rho$ visits only vertices that are labeled by $p$ at step $k$. We first give the definition of the Remove-Adjust procedure and then give some intuition about it.

\begin{defi}[Remove-Adjust procedure] 
\label{def:remove}

Let $k \in \mathbb{N} \setminus \{0\}$.
\begin{itemize}
\item If $k$ is odd, process the \emph{Remove} operation: 
\begin{itemize}
\item If for some $v \in V_i$ there exists $p \in \rP_{k-1}(v)$ and $v' \in \Succ(v)$ such that $p_i < p'_i$ for all $p' \in \rP_{k-1}(v')$,
then $\rP_k(v) = \rP_{k-1}(v) \backslash \{p\}$ and for all $u \neq v$, $\rP_k(u) = \rP_{k-1}(u)$. 
\item If such a vertex $v$ does not exist, then $\rP_k(u) = \rP_{k-1}(u)$ for all $u \in V$.
\end{itemize}
\item If $k$ is even, process the \emph{Adjust} operation: 
\begin{itemize}
\item If some payoff $p$ was removed from $\rP_{k-2}(v)$ (that is, $\rP_{k-1}(v) = \rP_{k-2}(v)\setminus \{p\}$), then 
\begin{itemize}
\item For all $u \in V$ such that $p \in \rP_{k-1}(u)$, check whether there still exists a $(p,k-1)$-labeled play with payoff $p$ from $u$. If it is the case, then $\rP_{k}(u) = \rP_{k-1}(u)$, otherwise $\rP_{k}(u) = \rP_{k-1}(u) \setminus \{p \}$.
\item For all $u \in V$ such that $p \notin \rP_{k-1}(u)$:
$\rP_{k}(u) = \rP_{k-1}(u).$
\end{itemize}
\item Otherwise $\rP_k(u) = \rP_{k-1}(u)$ for all $u \in V$.
\end{itemize}
\end{itemize}
\end{defi}

Let us explain the Remove operation. Let $p$ that labels vertex $v$. This means that it is the payoff of a potential subgame outcome of a weak SPE that starts in $v$. Suppose that $v$ is a vertex of player~$i$ and $v$ has a successor $v'$ such that $p_i < p'_i$ for all $p'$ labeling $v'$. Then $p$ cannot be the payoff of $\rest{\sigma}{h}$ in the subgame $(\rest{\mathcal{G}}{h},v)$ for some weak SPE $\sigma$ and some history $h$, otherwise player~$i$ would have a profitable (one-shot) deviation by moving from $v$ to $v'$ in this subgame.

Now it may happen that for another vertex $u$ having $p$ in its labeling, all potential subgame outcomes of a weak SPE from $u$ with payoff $p$ necessarily visit vertex $v$. As $p$ has been removed from the labeling of $v$, these potential plays do no longer survive and $p$ is also removed from the labeling of $u$ by the Adjust operation.

We can state the existence of a fixpoint of the sequences $(\rP_k(v))_{k\in \mathbb{N}}$, $v \in V$, in the following meaning:
 
 \begin{prop}[Existence of a fixpoint]
 There exists an even natural number $k^* \in \mathbb{N}$ such that for all $v \in V$, $\rP_{k^*}(v) = \rP_{k^*+1}(v) = \rP_{k^*+2}(v)$. 
 \end{prop} 
 \begin{proof}
For all $v \in V$, the sequence $(\rP_k(v))_{k\in \mathbb{N}}$ is nonincreasing because the Remove and Ajdust operations never add a new payoff. As each $\rP_0(v)$ is finite (it contains at most $2^{|\Pi|}$ payoffs), there exists a natural odd number $k^*+1$ such that for all $v \in V$, $\rP_{k^*}(v) = \rP_{k^*+1}(v)$ during the Remove operation, and thus for all $v \in V$, $\rP_{k^*+1}(v) = \rP_{k^*+2}(v)$ during the Adjust operation.
 \end{proof}

\begin{example}
	\label{exemple:pointFixe}
	We illustrate the different steps of the Remove-Adjust procedure on the example depicted in Figure~\ref{fig:weakNEvsNE}, and we display the result of this computation in Table~\ref{example:remove-adjustProc}. Initially, the sets $\rP_0(v)$, $v \in V$, contains all payoffs $p$ such that there exists a play $\rho \in \Plays(v)$ with $\Gain(\rho) = p$. 
	At step $k = 1$, we apply a Remove operation to $v = v_4$ (this is the only possible $v$): $v$ is a vertex of player~$i = 2$ and $v$ has a successor $v' = v_5$ such that $(0,1) \in \rP_0(v_5)$. Therefore $(0,0)$ is removed from $\rP_0(v_4)$ to get $\rP_1(v_4)$. By definition of the Remove operation, the other sets $\rP_0(u)$ are not modified and are thus equal to $\rP_1(u)$.
 	At step $k=2$, we apply an Adjust operation. The only way to obtain the payoff $(0,0)$ from $v_0$ is by visiting $v_4$ with the play $v_0v_4v_6^\omega$. As there does not exist a $((0,0),1)$-labeled play with payoff $(0,0)$ anymore, we have to remove $(0,0)$ from $\rP_1(v_0)$. The other sets $\rP_1(v)$ remain unchanged. 
	At step $k=3$, the Remove operation removes payoff $(1,0)$ from $\rP_2(v_0)$ due to the unique payoff $(0,1)$ in $\rP_2(v_4)$.
	At step $k=4$, the Adjust operation leaves all sets $\rP_3(v)$ unchanged. 
	Finally at step $k=5$, the Remove operation also leaves all sets $\rP_4(v)$ unchanged, and the fixpoint is reached.
	
	\begin{table}[h!]
		\caption{Computation of the fixpoint on the example of Figure~\ref{fig:weakNEvsNE}}
		\label{example:remove-adjustProc}
		\centering
		\scalebox{0.8}{
		\begin{tabular}{|l||l|l|l|l|l|l|l|}
			\hline
			& $v_0$ & $v_1$ & $v_2$ & $v_3$ & $v_4$ & $v_5$ & $v_6$\\
			\hline
			\hline
			$\rP_0(v)$ & $\{ (0,0),(1,0),(0,1) \}$ & $\{(1,0),(0,1)\}$ &$\{(1,0),(0,1)\}$ &$\{(0,1)\}$ & $\{\mathbf{(0,0)},(0,1)\}$&$\{(0,1)\}$& $\{(0,0)\}$\\
			\hline
			$\rP_1(v)$ & $\{ \textbf{(0,0)},(1,0),(0,1) \}$ & $\{(1,0),(0,1)\}$ &$\{(1,0),(0,1)\}$ &$\{(0,1)\}$ & $\{(0,1)\}$&$\{(0,1)\}$& $\{(0,0)\}$\\
			\hline 
			$\rP_2(v)$ & $\{\textbf{(1,0)},(0,1) \}$ & $\{(1,0),(0,1)\}$ &$\{(1,0),(0,1)\}$ &$\{(0,1)\}$ & $\{(0,1)\}$&$\{(0,1)\}$& $\{(0,0)\}$\\
			\hline
			$\rP_3(v)$  & $\{(0,1) \}$ & $\{(1,0),(0,1)\}$ &$\{(1,0),(0,1)\}$ &$\{(0,1)\}$ & $\{(0,1)\}$&$\{(0,1)\}$& $\{(0,0)\}$\\
			\hline
			$\rP_4(v) = \rP_{k^*}(v)$  & $\{(0,1) \}$ & $\{(1,0),(0,1)\}$ &$\{(1,0),(0,1)\}$ &$\{(0,1)\}$ & $\{(0,1)\}$&$\{(0,1)\}$& $\{(0,0)\}$\\
			\hline
		\end{tabular}}
	\end{table}
\end{example}

\subsection{Characterization and good symbolic witness}
\label{sec: witness}

The fixpoint $\rP_{k^*}(v)$, $v \in V$, provides a characterization of
the payoffs of all weak SPEs as described in the following
theorem. This result is in the spirit of the classical \emph{Folk
    Theorem} which characterizes the payoffs of all NEs in infinitely
  repeated games (see for instance~\cite[Chapter~8]{fudenberg1991game}).

\begin{thm}[Characterization]
 \label{folkThm}
 Let $(\mathcal{G},v_0)$ be an initialized Boolean game with prefix-independent gain functions. Then there exists a weak SPE $\sigma$ with payoff $p_0$ in $(\mathcal{G},v_0)$ if and only if $\rP_{k^*}(v) \neq \emptyset$ for all $v \in \Succ^*(v_0)$ and $p_0 \in \rP_{k^*}(v_0)$.\footnote{We use notation $p_0 \in \{0,1\}^{|\Pi|}$ to highlight that this is the payoff of $\sigma$ from vertex $v_0$. It should not be confused with any component $p_i$, $i \in \Pi$, of a payoff $p$.}
 \end{thm}

In this theorem, only sets $\rP_{k^*}(v)$ with $v \in \Succ^*(v_0)$ are considered. Indeed subgames $(\rest{\mathcal{G}}{h},v)$ of $(\mathcal{G},v_0)$ deals with histories $hv \in \Hist(v_0)$, that is, with vertices $v$ reachable from $v_0$. The rest of this section is devoted to the proof of Theorem~\ref{folkThm}. 

We begin with a lemma that states that if a payoff $p$ survives at step $k$ in the labeling of $v$, this means that there exists a play with payoff $p$ from $v$ that only visits vertices also labeled by $p$.

\begin{lemma}
 \label{intermediaryResults:lemmaLabeledPlay}
 For all even $k$ and in particular for $k = k^*$, $p$ belongs to $\rP_k(v)$ if and only if there exists a $(p,k)$-labeled play $\rho \in \Plays(v)$ such that $\Gain(\rho) = p$.
 \end{lemma}

 \begin{proof}
 $(\Leftarrow)$ Suppose that there exists a $(p,k)$-labeled play $\rho = \rho_0\rho_1 \ldots \in \Plays(v)$ such that $\Gain(\rho) = p$. By definition of a $(p,k)$-labeled play, we have $p \in \rP_{k}(\rho_j)$ for all $j$, and so in particular for $j = 0$.
 
 $(\Rightarrow)$ Let us prove that if $p$ belongs to $\rP_k(v)$, then there exists a $(p,k)$-labeled play $\rho \in \Plays(v)$ such that $\Gain(\rho) = p$. We proceed by induction on $k$. For $k=0$, the assertion is satisfied by definition of $\rP_0(v)$ and because $\Gain_i$ is prefix-independent for all $i\in \Pi$. 
 
Suppose that the assertion is true for an even $k$ and let us prove that it remains true for $k+2$. Let $p \in \rP_{k+2}(v)$. As $\rP_{k+2}(v) \subseteq \rP_{k+1}(v) \subseteq \rP_{k}(v)$, we have $p \in \rP_{k}(v)$ and there exists a $(p,k)$-labeled play $\rho \in \Plays(v)$ such that $\Gain(\rho) = p$ by induction hypothesis. In other words $p \in \rP_{k}(\rho_j)$ for all $j$.  

We suppose that there exists $v’$ such that $\rP_{k+2}(v') \neq \rP_{k}(v')$ (the fixpoint is not reached), otherwise $p \in \rP_{k+2}(\rho_j)$ for all $j$ and $\rho$ is also a $(p,k+2)$-labeled play. Therefore the Remove operation has removed some payoff $p'$ from one $\rP_{k}(v')$ and the Adjust operation has possibly removed $p'$ from some other $\rP_{k}(u)$. If $p' \neq p$, then clearly $p$ still belongs to each $\rP_{k+2}(\rho_j)$ and $\rho$ is again a $(p,k+2)$-labeled play. If $p' = p$, then $v' \neq v$ since $p \in \rP_{k+2}(v)$ by hypothesis. Moreover, by the Adjust operation, this means that there exists a $(p,k+1)$-labeled play $\pi = \pi_0\pi_1\ldots$ with payoff $p$ from $v$ which never visits $v'$. Let us show that $\pi$ is also a $(p,k+2)$-labeled play, that is, $p \in \rP_{k+2}(\pi_j)$ for all $j$. Each suffix $\pi_j\pi_{j+1} \ldots$ of $\pi$ is a $(p,k+1)$-labeled play with payoff $p$ thanks to prefix-independence of $\Gain$. By the Adjust operation, it follows that $\rP_{k+2}(\pi_j) = \rP_{k+1}(\pi_j)$ for all $j$. This concludes the proof.
%
%
 \end{proof} 
 
The proof of Theorem~\ref{folkThm} uses the concept of (good) symbolic witness defined hereafter. Some intuition about it is given after the definitions.

\begin{defi}[Symbolic witness] 
	\label{def:symbolicWitnesses}
	 Let $(\mathcal{G},v_0)$ be an initialized Boolean game with prefix-independent gain functions. Let $I \subseteq (\Pi \cup \{0\}) \times V$ be the set 
	 $$I = \{(0,v_0)\} ~\cup~ \{ (i,v') \mid \text{ there exists } (v,v')\in E \text{ such that } v, v' \in \Succ^*(v_0) \text{ and } v \in V_i \}.$$
	 A  \emph{symbolic witness} is a set $\Wit = \{\rho_{i,v} \mid (i,v) \in I \}$ such that each $\rho_{i,v} \in \Wit$ is a lasso of $G$ with $\First(\rho_{i,v}) = v$ and with length bounded by $2 \cdot |V|^2$.
\end{defi}

A symbolic witness has thus at most $|V| \cdot |\Pi| + 1$ lassoes (by definition of $I$) with polynomial length.

\begin{defi}[Good symbolic witness]
	\label{def:Good}
	A symbolic witness $\Wit$ is \emph{good} if for all $\rho_{j,u}, \,\, \rho_{i,v'} \in \Wit$, for all vertices $v \in \rho_{j,u}$ such that $v \in V_i$ and $v' \in \Succ(v)$, we have 
	 $\Gain_i(\rho_{j,u}) \geq \Gain_i(\rho_{i,v'})$.
\end{defi}

The condition of Definition~\ref{def:Good} is depicted in Figure~\ref{fig:Good}.

\begin{figure}[h!]
	\centering
	\begin{tikzpicture}
		\node[draw,circle,minimum width=20pt](u) at (0,0){$u$};
		\node(point1) at (1.5,0){$\ldots$};
		\node[draw,circle,minimum width=20pt](v) at (3,0){$v$};
		\node(vi) at (3,-0.7){$\in V_i$};
		\node[draw,circle,minimum width=20pt](v') at (3.5,1){$ v'$};
		\node(point2) at (5,1){$\ldots$};
		\node[draw,circle,minimum width=20pt](blanc1) at (6.5,1){}; 
		\node(point3) at (8,1){$\ldots$};
		\node[draw,circle,minimum width=20pt](blanc2) at (9.5,1){}; 
		\node (path1) at (10.5,1){$\rho_{i,v'}$};
		\node(point4) at (5,0){$\ldots$};
		\node[draw,circle,minimum width=20pt](blanc3) at (6.5,0){}; 
		\node (point5) at (8,0){$\ldots$};
		\node[draw,circle,minimum width=20pt](blanc4) at (9.5,0){};
		\node (path2) at (10.5,0){$\rho_{j,u}$};	
		
		\draw[->] (u) -- (1,0);
		\draw[->] (2,0) -- (v);
		\draw[->] (v) -- (v');
		\draw[->] (v') -- (4.5,1);
		\draw[->] (5.5,1) --(blanc1);
		\draw[->] (blanc1) -- (7.5,1);
		\draw[->] (8.5,1) -- (blanc2);
		\draw[->] (blanc2) to [bend right] (blanc1);
		
		\draw[->] (v) -- (4.5,0);
		\draw[->] (5.5,0) -- (blanc3);
		\draw [->] (blanc3) -- (7.5,0);
		\draw[->] (8.5,0) -- (blanc4);
		\draw[->] (blanc4) to [bend left] (blanc3);
	\end{tikzpicture}
	\caption{The condition of Definition~\ref{def:Good}}
	\label{fig:Good}
\end{figure}
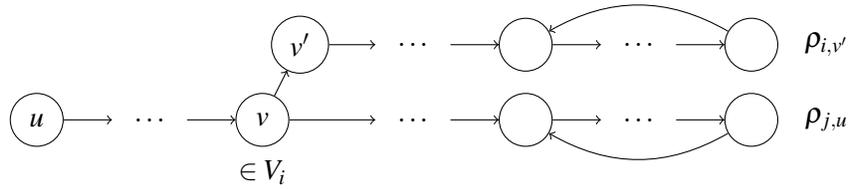

Let us now give some intuition. A strategy profile $\sigma$ in $(\mathcal{G},v_0)$ induces an infinite number of subgame outcomes $\outcome{\rest{\sigma}{h}}{v}$, $hv \in \Hist(v_0)$. A symbolic witness $\Wit$ is a \emph{compact} representation of $\sigma$. It is a finite set of lassoes that represent some subgame outcomes of $\sigma$: the lasso $\rho_{0,v_0}$ of $\Wit$ represents the outcome $\outcome{\sigma}{v_0}$, and its other lassoes $\rho_{i,v'}$ represents the subgame outcome $\outcome{\rest{\sigma}{h}}{v'}$ for some particular histories $hv' \in \Hist(v_0)$. The index $i$ records that player~$i$ can move from $v$ (the last vertex of $h$) to $v'$ (with the convention that $i = 0$ for the outcome $\outcome{\sigma}{v_0}$). When $\sigma$ is a weak SPE, the related symbolic witness $\Wit$ is good, that is, its lassoes avoid profitable one-shot deviations between them. 

\begin{example} 
We come back to our running example. The weak SPE of Example~\ref{ex:weakSPE} depicted in Figure~\ref{fig:weakNEvsNE} has payoff $p = (0,1)$. A symbolic witness $\Wit$ of $\sigma$ is given in Table~\ref{table:setOfSymbolicWitnesses} which is here composed of all the subgame outcomes of $\sigma$. One can check that $\Wit$ is a good symbolic witness. For instance, consider its lassoes $\rho_{0,v_0} = v_0v_1v_2v_3^\omega$ and $\rho_{1,v_1} = v_1v_2v_3^\omega$, the vertex $v_2 \in V_1$ of $\rho_{0,v_0}$ and the edge $(v_2,v_1)$. We have $\Gain_1(\rho_{0,v_0}) \geq \Gain_1(\rho_{1,v_1})$. Indeed in the subgame $(\rest{\mathcal{G}}{v_0v_1},v_2)$, player~$1$ has no profitable one-shot deviation by using the edge $(v_2,v_1)$.
\begin{table}[h!]
	\centering
	\caption{An example of good symbolic witness}
	\begin{tabular}{|l|l|l|l|l|l|l|l|l|l|}
		\hline
		$(i,v)$ & $(0,v_0)$ &$(2,v_4)$ & $(1,v_2)$ & $(1,v_1)$ & $(1,v_3)$ & $(2,v_5)$ & $(2,v_6)$ & $(1,v_5)$ & $(1,v_6)$ \\
		\hline
		lasso & $v_0v_1v_2v_3^\omega$ & $v_4v_5^\omega$ & $v_2v_3^\omega$ & $v_1v_2v_3^\omega$ &  $v_3^\omega$ &  $v_5^\omega$  & $v_6^\omega$ & $v_5^\omega$ & $v_6^\omega$ \\
		\hline
		payoff & $(0,1)$ & $(0,1)$ & $(0,1)$ & $(0,1)$ & $(0,1)$& $(0,1)$ & $(0,0)$ & $(0,1)$ & $(0,0)$\\
		\hline
	\end{tabular}
		\label{table:setOfSymbolicWitnesses}
\end{table}
\end{example}


In Proposition~\ref{prop:lasseFolkThm} below, we are going to prove that there exists a weak SPE if and only if there exists a good symbolic witness, and that the existence of this witness is equivalent to the non-emptiness of the fixpoint $\rP_{k^*}(v)$, $v \in V$. In this way, we will prove Theorem~\ref{folkThm}. We will see that the lassoes $\rho_{i,v}$ of a good symbolic witness can be constructed from $(p,k^*$)-labeled plays for well-chosen payoffs $p \in \rP_{k^*}(v)$.

%
%
%
%
%
%
%

\begin{prop}
	\label{prop:lasseFolkThm}
	 Let $(\mathcal{G},v_0)$ be an initialized Boolean game with prefix-independent gain functions. The following assertions are equivalent:
	\begin{enumerate}
		\item There exists a weak SPE with payoff $p_0$ in $(\mathcal{G},v_0)$; \label{Lasse:1}
		\item $\rP_{k^*}(v) \neq \emptyset$ for all $v \in \Succ^*(v_0)$ and $p_0 \in \rP_{k^*}(v_0);$ \label{Lasse:2}
		\item There exists a good symbolic witness $\Wit$ that contains a lasso $\rho_{0,v_0}$ with payoff $p_0$;  \label{assert:existenceWitnesses}
		\item There exists a finite-memory weak SPE $\sigma$ with payoff $p_0$ in $(\mathcal{G},v_0)$ such that the size of each strategy $\sigma_i$ is in $\mathcal{O}(|V|^3 \cdot |\Pi|)$. \label{assert:existenceWSPEpoly}
	\end{enumerate}	
\end{prop}

\begin{proof}
We prove that $1 \Rightarrow 2 \Rightarrow 3 \Rightarrow 4 \Rightarrow 1$. 

$(1 \Rightarrow 2)$ Suppose that there exists a weak SPE $\sigma$ with payoff $p_0$ in $(\mathcal{G},v_0)$. To show that $\rP_{k^*}(v) \neq \emptyset$ for all $v \in \Succ^*(v_0)$, let us prove by induction on $k$ that 
\begin{align} \label{eq:nonempty}
\Gain(\outcome{\rest{\sigma}{h}}{v}) \in \rP_k(v) \text{ for all } hv \in \Hist(v_0).
\end{align}
For case $k = 0$, this is true by definition of $\rP_0(v)$. 
Suppose that this assertion is satisfied for an even $k$. Let us prove that it remains true for $k+2$ by showing that payoff $p =\Gain(\outcome{\rest{\sigma}{h}}{v}) \in \rP_k(v)$ can be removed neither from $\rP_{k}(v)$ at step $k+1$, nor from $\rP_{k+1}(v)$ at step $k+2$. 
	 \begin{itemize}
	 \item Payoff $p$ cannot be removed from $\rP_{k}(v)$ by the Remove operation at step $k+1$. Otherwise, if $v \in V_i$, this means that there exists $v'\in \Succ(v)$ such that 
	 \begin{align}\forall p'\in \rP_{k}(v'),\, p_i < p'_i. \label{eq1}\end{align}

	 By induction hypothesis, \begin{align} \Gain(\outcome{\rest{\sigma}{hv}}{v'}) \in \rP_{k}(v'). \label{eq2}\end{align}
	To get a contradiction, we prove that in the subgame $(\rest{\mathcal{G}}{h},v)$ there exists a one-shot deviating strategy $\sigma'_i$ from $\sigma_{i\restriction h}$ that is a profitable deviation for player~$i$. We define $\sigma'_i$ that only differs from $\sigma_{i\restriction h}$ on $v$: $\sigma'_i(v) = v'$. Therefore we get $\Gain(h \outcome{\sigma'_i, \sigma_{-i\restriction h}}{v}) = \Gain (hv \outcome{\rest{\sigma}{hv}}{v'})$. It follows by \eqref{eq1}, \eqref{eq2}, and prefix-independence of $\Gain_i$ that $\Gain_i( h \outcome{\rest{\sigma}{h}}{v}) = p_i < p'_i = \Gain_i (hv \outcome{\rest{\sigma}{hv}}{v'}) = \Gain(h \outcome{\sigma'_i, \sigma_{-i\restriction h}}{v})$. 
%
%
	 This is impossible since $\sigma$ is a weak SPE.

	\item Payoff $p$ cannot be removed from $\rP_{k+1}(v)$ by the Adjust operation at step $k+2$. Otherwise, this means that there exists $u$ such that $\rP_{k+1}(u) = \rP_{k}(u) \setminus \{p\}$ (by the Remove operation at step $k+1$) and there is no $(p,k+1)$-labeled play with payoff $p$ from $v$. However by Lemma~\ref{intermediaryResults:lemmaLabeledPlay}, as $p \in \rP_k(v)$, there exists a $(p,k)$-labeled play $\pi$ with payoff $p$ from $v$. This means that $\pi$ visit $u$. Let $h'u \in Hist(v)$ such that $h'u < \pi$. Then we get a contradiction with $\sigma$ being a weak SPE if we repeat the argument done before in the previous item for $u$ and the subgame $(\rest{\mathcal{G}}{hh'},u)$ (instead of $v$ and $(\rest{\mathcal{G}}{h},v)$).
	
	 \end{itemize}
	
Now that we know that $\rP_{k^*}(v) \neq \emptyset$ for all $v \in \Succ^*(v_0)$, it remains to prove that $p_0 \in \rP_{k^*}(v_0)$. By~\eqref{eq:nonempty}, we have $p_0 = \Gain(\outcome{\sigma}{v_0}) \in \rP_{k^*}(v_0)$.

$(2 \Rightarrow 3)$ Let us show how to build a good symbolic witness $\Wit$ from the non-empty fixpoint $\rP_{k^*}(v)$, $v \in V$. First recall that if $p \in \rP_{k^*}(v)$, then by Lemma~\ref{intermediaryResults:lemmaLabeledPlay} there exists a $(p,k^*)$-labeled play with payoff $p$ from $v$. Notice that such a play can be supposed to be a lasso with length at most $2 \cdot |V|^2$. Indeed it is proved in~\cite[Proposition 3.1]{BouyerBMU15} that given a play~$\rho$, one can construct a lasso $\rho'$ of length bounded by $2 \cdot |V|^2$ such that $\First(\rho) = \First(\rho')$, $\Occ(\rho) = \Occ(\rho')$, and $\Inf(\rho) = \Inf(\rho')$ (this construction eliminates some cycles of $\rho$ in a clever way). Therefore, if $\rho$ is a $(p,k^*)$-labeled play with payoff $p$ from $v$, the lasso $\rho'$ is also a $(p,k^*)$-labeled play with payoff $p$ from $v$. The required set $\Wit$ will be composed of some of these lassoes.

We start with $\Wit = \emptyset$. As $p_0 \in \rP_{k^*}(v_0)$, then there exists a $(p_0,k^*)$-labeled lasso $\rho_{0,v_0}$ with payoff $p_0$ from $v_0$ that we add to $\Wit$. For all $v, v' \in \Succ^*(v_0)$ such that $v \in V_i$ and $v' \in Succ(v)$, let $p'$ be a payoff in $\rP_{k^*}(v')$ such that 
\begin{align}
p'_i = \Min\{ q_i \mid q \in \rP_{k^*}(v') \}. \label{eq3}
\end{align} 
This payoff exists since $\rP_{k^*}(v') \neq \emptyset$ by hypothesis. Then there exists a $(p',k^*)$-labeled lasso $\rho_{i,v'}$ with payoff $p'$ from $v'$ that we add to $\Wit$.

This set $\Wit$ is a symbolic witness by construction. It remains to prove that it is good. Let $v \in \rho_{j,u}$ such that $v \in V_i$ and $\rho_{j,u} \in \Wit$. As $\rho_{j,u}$ is a $(p,k^*)$-labeled lasso for some payoff $p$, we have $p \in \rP_{k^*}(v)$. Furthermore, as $\rP_{k^*}(v) =  \rP_{k^*+1}(v)$ (by the fixpoint), this means that $p$ was not removed from $\rP_{k^*}(v)$ by the Remove operation at step $k^*$. In particular, by definition of the payoff $p'$ of $\rho_{i,v'}$ (see \eqref{eq3}), we have $p_i \geq p'_i$, that is $\Gain_i(\rho_{j,u}) \geq \Gain_i(\rho_{i,v'})$. This shows that $\Wit$ is a good symbolic witness.
	

\medskip
$(3 \Rightarrow 4)$ Let $\Wit = \{\rho_{i,v} \mid (i,v) \in I \}$ be a good symbolic witness that contains a lasso $\rho_{0,v_0}$ with payoff $p_0$. We define a strategy profile $\sigma$ step by step by induction on the subgames of $(\mathcal{G},v_0)$. We first partially build $\sigma$ such that $\outcome{\sigma}{v_0} = \rho_{0,v_0}$. Consider next $hvv' \in \Hist(v_0)$ with $v \in V_i$ such that $\outcome{\rest{\sigma}{h}}{v}$ is already built but not $\outcome{\rest{\sigma}{hv}}{v'}$. Then we extend the definition of $\sigma$ such that 
\begin{align}
\outcome{\rest{\sigma}{hv}}{v'} = \rho_{i,v'}. \label{eq4}
\end{align} 
Notice that $\outcome{\rest{\sigma}{h}}{v}$ being already built means that there exists $h' \leq h$ and $(j,u) \in I$ such that
\begin{align}
h' \outcome{\rest{\sigma}{h'}}{u} =  h'\rho_{j,u} = h\outcome{\rest{\sigma}{h}}{v}. \label{eq5}
\end{align} 

%

		Let us prove that $\sigma$ is a very weak SPE (and so a weak SPE by Proposition~\ref{prop:equivWSPEetVWSPE}). Consider the subgame $(\rest{\mathcal{G}}{h},v)$ and the one-shot deviating strategy $\sigma'_i$ from $\sigma_{i\restriction h}$ such that $\sigma'_i(v) = v'$. We have to prove that 
\begin{align}		
	\Gain_i(h \outcome{\rest{\sigma}{h}}{v}) \geq \Gain_i (hv \outcome{\rest{\sigma}{hv}}{v'}). \label{eq6}
\end{align} 
By \eqref{eq4}, \eqref{eq5}, and prefix independence of $\Gain_i$, we have 
\begin{align*} 
\Gain_i (hv \outcome{\rest{\sigma}{hv}}{v'}) &= \Gain_i({\rho_{i,v'}}),&  \\
\Gain_i(h \outcome{\rest{\sigma}{h}}{v}) &= \Gain_i(\rho_{j,u}).& 
\end{align*}
Inequality~\eqref{eq6} follows from these equalities and the fact that $\Wit$ is a good symbolic witness (see Figure~\ref{fig:Good}). 

Notice that $\sigma$ has payoff $p_0$ by construction ($\rho_{0,v_0}$ has payoff $p_0$). It remains to show that $\sigma$ is finite-memory. Having $(j,u)$ in memory (the last deviating player $j$ and the vertex $u$ where he moved), the Moore machines $\mathcal{M}_i$, $i \in \Pi$, representing each strategy $\sigma_i$, have to produce together the lasso $\rho_{j,u}$ of length bounded by $2 \cdot |V|^2$. As $|{\mathcal P}| = |I| \leq |V| \cdot |\Pi| + 1$, the size of each $\sigma_i$ is in $\mathcal{O}(|\Pi|\cdot |V|^3)$.
		
%
%
%
%
%
%
%
%
%
%
%
%
%
%

$(4 \Rightarrow 1$) This implication is obvious.
\end{proof}

\begin{example} We consider again the running example. Thanks to Theorem~\ref{folkThm}, payoff $p_0 = (0,1)$ is the only possible payoff for a weak SPE in $(\mathcal{G},v_0)$ since $\rP_{k^*}(v_0)= \{(0,1)\}$ (see Example~\ref{exemple:pointFixe}). Let us illustrate the constructions of the proof of Proposition~\ref{prop:lasseFolkThm}.	
We build a good symbolic witness $\Wit$ as follows. First, we choose a $(p_0,k^*)$-labeled lasso $\rho_{0,v_0}$ with payoff $p_0$ from $v_0$: $\rho_{0,v_0} = v_0v_1v_2v_3^\omega$ (we could also have chosen $v_0v_4v_5^\omega$). Then let $v = v_1, v' = v_2 \in \Succ^*(v_0)$ such that $v \in V_1$ and $v' \in \Succ(v)$. As $\rP_{k^*}(v_2) = \{(1,0),(0,1)\}$, by \eqref{eq3}, we choose a $(p',k^*)$-labeled lasso $\rho_{1,v_2}$ with payoff $p' = (0,1)$ from $v_2$, for instance, $\rho_{1,v_2} = v_2v_3^\omega$.
The whole set $\Wit$ is depicted in Table~\ref{table:setOfSymbolicWitnesses}. From $\Wit$, we get the weak SPE of Example~\ref{ex:weakSPE} depicted in Figure~\ref{fig:weakNEvsNE}.
\end{example}

Theorem~\ref{folkThm} directly follows from Proposition~\ref{prop:lasseFolkThm}. Moreover this proposition highlights a property that will be very useful in Section~\ref{sec:classes} (the equivalence between the existence of a weak SPE and the existence of a good symbolic witness). Finally it shows that when the gain fonctions are prefix-independent, if there exists a weak SPE with a given payoff, then there always exists one with the same payoff but with strategies of polynomial size. We prove in Section~\ref{sec:cor} that for Reachability and Safety objectives which are not prefix-independent, we have the same result however with strategies of exponential size.

\begin{corollary}
\label{cor:finitememory}
Let $(\mathcal{G},v_0)$ be an initialized Boolean game. There exists a weak SPE in $(\mathcal{G},v_0)$ if and only if there exists a finite-memory weak SPE $\sigma$ with the same payoff. Moreover, the size of each strategy $\sigma_i$ is 
\begin{itemize}
\item in $\mathcal{O}(|V|^3 \cdot |\Pi|)$ for B\"uchi, Co-B\"uchi, Parity, Explicit Muller, Muller, Rabin, and Streett objectives,
\item in $\mathcal{O}(|V|^3 \cdot |\Pi| \cdot 2^{3\cdot |\Pi|} )$ for Reachability and Safety objectives.
\end{itemize}

\end{corollary}

\section{Complexity classes of the constraint problem}
\label{sec:classes}

In this section, we study the complexity classes of the constraint problem for Boolean games with classical $\omega$-regular objectives, except the case of Explicit Muller objectives that is postponed to Section~\ref{sec:ExplicitMuller} (see Table~\ref{tab:complexity}). The concept of good symbolic witness is essential in this study.

\subsection{NP-completeness}
\label{sec:NP}

We first prove that the constraint problem for co-Büchi, Parity, Muller, Rabin, and Streett objectives is NP-complete, and that it is in NP for B\"uchi objectives. 

\begin{thm}
\label{thm:NP}
The constraint problem for Boolean games with co-Büchi, Parity, Muller, Rabin, and Streett objectives is NP-complete. It is in NP for B\"uchi objectives. 
\end{thm}

\begin{proof}
We begin with the NP-easyness. The objectives considered in Theorem~\ref{thm:NP} are prefix-independent. We can thus apply Proposition~\ref{prop:lasseFolkThm}. Given thresholds $x, y \in \{0,1\}^{|\Pi|}$, there exists a weak SPE in $(\mathcal{G},v_0)$ with payoff $p$ such that $x \leq p \leq y$ if and only if there exists a good symbolic witness $\Wit$ that contains a lasso $\rho_{0,v_0}$ with payoff $p$. Hence a nondeterministic polynomial algorithm works as follows: guess a set $\Wit$ composed of at most $|\Pi| \cdot |V| + 1$ lassoes of length at most $2 \cdot |V|^2$ and check that $\Wit$ is a good symbolic witness that contains a lasso $\rho_{0,v_0}$ with payoff $p$ such that $x \leq p \leq y$. Clearly checking that $\Wit$ is a symbolic witness can be done in polynomial time. Checking that it is good requires to compute the payoffs of its lassoes and to compare them. This can also be done in polynomial time for B\"uchi, co-Büchi, Parity, Muller, Rabin, and Streett objectives.

We now proceed to the NP-hardness. It is obtained thanks to a polynomial reduction from SAT. In~\cite{Ummels08} is provided a polynomial reduction from SAT to the constraint problem for NEs in Boolean games with co-Büchi objectives. Due to the structure of the game constructed in this approach, the same reduction holds for the constraint problem for weak SPEs. As co-Büchi objectives can be polynomially translated into Parity, Muller, Rabin, and Streett objectives (see \cite{2001automata}), the constraint problem for Boolean games with those objectives is also NP-hard.
\end{proof}

\subsection{PSPACE-completeness}

In this section, we show that the constraint problem for Reachability and Safety objectives is PSPACE-complete.  

\begin{thm}
\label{thm:PSPACE}
The constraint problem for Boolean games with Reachability and Safety objectives is PSPACE-complete.
\end{thm}

Recall that weak SPEs and SPEs are equivalent notions for Reachability objectives (Proposition~\ref{prop:reach}). It follows from Theorem~\ref{thm:PSPACE} that the constraint problem for SPEs (instead of weak SPEs) for Boolean games with Reachability objectives is PSPACE-complete. We will see later (in Section~\ref{sec:cor}, from the proof of Theorem~\ref{thm:PSPACE}) that the constraint problem for SPEs is also PSPACE-complete for Safety objectives. 

\begin{corollary}
\label{cor:PSPACE}
The constraint problem for SPEs in Boolean games with Reachability and Safety objectives is PSPACE-complete.
\end{corollary}

We detail the proof of Theorem~\ref{thm:PSPACE} in the next two sections for Reachability objectives, and we also show how to adapt it for Safety objectives. To get the PSPACE-easyness, we transform the Boolean game $(\mathcal{G},v_0)$ with Reachability objectives (which are not prefix-independent) into a Boolean game $(\mathcal{G}',v'_0)$ with B\"uchi objectives. In this way, it is possible to use the concept of good symbolic witness as done before in Section~\ref{sec:NP}. Even if the size of the game $(\mathcal{G}',v'_0)$ is exponential in the size of the initial game $(\mathcal{G},v_0)$, we manage to get a PSPACE-membership thanks to the classical complexity result PSPACE $=$ APTIME.  The PSPACE-hardness is obtained with a polynomial reduction from QBF. The reduction is more involved than the one in Theorem~\ref{thm:NP}.  Indeed the reduction for NP-hardness already works for NEs whereas the reduction for PSPACE-hardness really exploits the subgame perfect aspects.

\subsubsection{PSPACE-easyness}
\label{reachPSPACEe}

We here prove that the constraint problem is in PSPACE for Reachability objectives, and we then explain how to adapt the proof for Safety objectives. 

\begin{prop}
\label{prop:PSPACE}
The constraint problem for Boolean games with Reachability objectives is in PSPACE.
\end{prop}

\begin{proof}
First, we transform the Boolean game $(\mathcal{G},v_0)$ with Reachability objectives into a Boolean game $(\mathcal{G}',v'_0)$ with B\"uchi objectives. This construction is classical, it stores inside the vertices the set of players who have already satisfied their objective. Suppose that in $(\mathcal{G},v_0)$, each player~$i$ aims at reaching a vertex of $F_i$, then we build $\mathcal{G}' = (\Pi, V', E', (V'_i)_{i \in \Pi}, (\Gain'_i)_{i\in \Pi})$ such that

\begin{itemize}
\item $V' = V \times 2^{\Pi}$;
\item $((v,I),(u,I')) \in E'$ if and only if $(v,u)\in E$ and $I' =  I  \cup \{i\in \Pi \mid v' \in F_i \}$;
\item $(v,I) \in V'_i$ if and only if $v \in V_i$;
\item $\Gain'_i$ corresponds to the Büchi objective $F'_i = \{ (v, I) \mid v \in V, i \in I \}$; and
\item $v'_0 = (v_0, I_0)$ is the initial vertex such that $I_0 = \{ i \in \Pi \mid v_0 \in F_i \}$.
\end{itemize}

Clearly there is a one-to-one correspondence between plays $\rho = v_1v_2 \ldots v_k \ldots$ in $\mathcal{G}$ from $v_1$ and $\rho' = (v_1,I_1)(v_2,I_2) \ldots (v_k,I_k) \ldots$ in $\mathcal{G}'$ from $(v_1,I_1)$, with the important property that
\begin{align}
I_k \subseteq I_{k+1} \text{ for all } k\geq 1. \label{eq:Ik}
\end{align} 
In particular, there exists a weak SPE with payoff $p_0$ in $(\mathcal{G},v_0)$ if and only if there is one in $(\mathcal{G}',v'_0)$ if and only if there is a good symbolic witness $\mathcal{P}$ containing a lasso $\rho_{0,v'_0}$ with payoff $p_0$ in $(\mathcal{G}',v'_0)$ (by Proposition~\ref{prop:lasseFolkThm}).

Second, let us study the lengths of the lassos $\rho_{i,v'}$, with $i \in \Pi \cup \{0\}$, $v' \in V'$, of $\mathcal{P}$. There exists a better bound than the bound $2 \cdot |V'|^2 = 2 \cdot |V|^2 \cdot 2^{2 \cdot |\Pi|}$ of Definition~\ref{def:symbolicWitnesses}: each lasso $\rho_{i,v'}$ of $\mathcal P$ has a polynomial length bounded by
\begin{align}
(|\Pi| + 1) \cdot |V|. \label{eq:lasso}
\end{align} 
Indeed recall that these lassoes are constructed from some $(p,k^*)$-labeled plays $\rho$ with payoff $p$ from $v'$ (see the proof of implication $(2 \Rightarrow 3)$ of Proposition~\ref{prop:lasseFolkThm}). We adapt this construction as follows. Let $\rho = (v_1,I_1) (v_2,I_2) \ldots (v_k, I_k) \ldots$ be such a play from $v'$. By \eqref{eq:Ik}, there exists $I \subseteq \Pi$ and $k \in \mathbb{N}$ such that for all $k' \geq k$, $I_{k'} = I$. Hence from $\rho$, we can construct a lasso $\rho'$ of length bounded by \eqref{eq:lasso} such that $\First(\rho') = \First(\rho)$, $\Occ(\rho') \subseteq \Occ(\rho)$, and $\Gain(\rho') = \Gain(\rho)$,
\begin{itemize}
\item by eliminating all cycles in the history $(v_1,I_1) (v_2,I_2) \ldots (v_{k-1}, I_{k-1})$ (leading to a history of length at most $|\Pi| \cdot |V|$), and
\item by detecting in the play $(v_k,I) (v_{k+1},I) \ldots$ the first repeated vertex $(v_{k'},I) = (v_{k'+\ell +1},I)$ and replacing this play by the lasso $(v_k,I) (v_{k+1},I) \ldots ((v_{k'},I) \ldots (v_{k'+\ell},I) )^\omega$ of length at most $|V|$. 
\end{itemize}
In this way, if $\rho$ is a $(p,k^*)$-labeled play with payoff $p$ from $v'$, then the constructed lasso $\rho'$ is also a $(p,k^*)$-labeled play with payoff $p$ from $v'$.

Third we prove PSPACE-membership of the constraint problem by proving that it is in APTIME. Given the game $(\mathcal{G}',v'_0)$ and two thresholds $x,y \in \{0,1\}^{|\Pi|}$, the alternating Turing machine works as follows. Existential and universal states (respectively controlled by player~$\vee$ and player~$\wedge$) alternate along an execution of the machine. Player~$\vee$ proposes a lasso $\rho_{j,u'}$ of length bounded by $(|\Pi| + 1) \cdot |V|$ (in the initial state, he proposes a lasso $\rho_{0,v'_0}$). Then Player~$\wedge$ chooses a vertex $w' \in V'_i$ of $\rho_{j,u'}$ and proposes to move to $v'$ such that $(w',v') \in E'$. Player~$\vee$ reacts by proposing a lasso $\rho_{i,v'}$ of length bounded by $(|\Pi| + 1) \cdot |V|$, and so on. The execution stops after
\begin{align}
2 \cdot |\Pi|^2\cdot |V| + 1 \text{ turns. } \label{eq:turns}
\end{align}
Such an execution is accepting if:
\begin{itemize}
\item for the payoff $p_0$ of the initial lasso $\rho_{0,v'_0}$, we have $x \leq p_0 \leq y$;
\item for each lasso $\rho_{j,u'}$ proposed by player~$\vee$, for the corresponding move $(w',v') \in E'$ with $w' \in V'_i$ made by player~$\wedge$, and the answer $\rho_{i,v'}$ of player~$\vee$, we have $\Gain_i(\rho_{j,u'}) \geq \Gain_i(\rho_{i,v'})$.
\end{itemize}

The intuition is that if there exists in $(\mathcal{G}',v'_0)$ a good symbolic witness $\mathcal{P}$ containing a lasso $\rho_{0,v'_0}$ with payoff $p_0$   such that $x \leq p_0 \leq y$, then player~$\vee$ will play with the lassoes of $\mathcal{P}$ according to Definition~\ref{def:Good}. Notice that along an execution of the Turing machine, player~$\wedge$ has no interest to choose twice the same pair $(i,v')$ since player~$\vee$ will react with the same lasso $\rho_{i,v'}$. Remembering property $\eqref{eq:Ik}$, the maximum number of times that player~$\wedge$ has to play is 
\begin{align}
|\Pi|^2\cdot |V|. 
\end{align}
Indeed for a fixed $I \subseteq \Pi$, player~$\wedge$ can choose at most $|\Pi| \cdot |V|$ different pairs $(i,v')$ with $v'$ of the form $(v,I)$, and the size of $I$ can only increase. This explains the number of turns of any execution of the machine (see \eqref{eq:turns}): an initial lasso proposed by player~$\vee$ followed by $|\Pi|^2\cdot |V|$ alternations between moves of both players~$\vee$ and~$\wedge$.

Checking whether an execution is accepting is done in polynomial time since player~$\vee$ proposes lassoes of polynomial size by \eqref{eq:lasso}, there is a polynomial numbers of turns by \eqref{eq:turns}, and computing and comparing payoffs of lassoes is done in polynomial time. So the constraint problem is in APTIME $=$ PSPACE.
\end{proof}

The constraint problem for Safety objectives is solved similarly by transforming the given Boolean game into one with co-B\"uchi (instead of B\"uchi) objectives.

\begin{prop}
\label{prop:PSPACEsafety}
The constraint problem for Boolean games with Safety objectives is in PSPACE.
\end{prop}

\subsubsection{PSPACE-hardness}

We now prove that the constraint problem is PSPACE-hard for Reachability objectives, and we then show how to adapt the proof for Safety objectives.

\begin{prop}
	\label{prop:reachPSPACE-hard}
	The constraint problem for Boolean games with Reachability objectives is PSPACE-hard.
\end{prop}

To prove this proposition, we give a polynomial reduction from the QBF problem that is PSPACE-complete. This problem is to decide whether a fully quantified Boolean formula $\psi$ is true. The formula $\psi$ can be assumed to be in prenex normal form $Q_1x_1Q_2x_2 \ldots Q_mx_m \,\phi(X)$ such that the quantifiers are alternating existential and universal quantifiers ($Q_1 = \exists$, $Q_2 = \forall$, $Q_3 = \exists, \ldots$), $X = \{ x_1, x_2, \ldots, x_m\}$ is the set of quantified variables, and $\phi(X) = C_1 \wedge \ldots \wedge C_n$ is an unquantified Boolean formula over $X$ equal to the conjunction of the clauses $C_1, \ldots, C_n$. 

Such a formula $\psi$ is true if there exists a value of $x_1$ such that for all values of $x_2$, there exists a value of $x_3$ $\ldots$, such that the resulting valuation $\nu$ of all variables of $X$ evaluates $\phi(X)$ to true. Formally, for each odd (resp. even) $k$, $1 \leq k \leq m$, let us denote by $f_k: \{0,1\}^{k-1} \rightarrow \{0,1\}$ (resp. $g_k: \{0,1\}^{k-1} \rightarrow \{0,1\}$) a valuation of variable $x_k$ given a valuation of previous variables $x_1, \ldots, x_{k-1}$\footnote{Notice that $f_1: \emptyset \rightarrow \{0,1\}$.}. Given theses sequences $f= f_1, f_3, \ldots$ and $g = g_2, g_4, \ldots$, let us denote by $\nu = \nu_{(f,g)}$ the valuation of all variables of $X$ such that $\nu(x_1) = f_1$, $\nu(x_2) = g_2(\nu(x_1))$, $\nu(x_3) = f_3(\nu(x_1)\nu(x_2))$, $\ldots$. Then 

\begin{center}
$\psi = Q_1x_1Q_2x_2 \ldots Q_mx_m \,\phi(X)$ is true \\
if and only if \\
there exist $f= f_1, f_3, \ldots$ such that for all $g = g_2,g_4, \ldots$, the valuation $\nu_{f,g}$ evaluates $\phi(X)$ to true.
\end{center} 

\begin{proof}[Proof of Proposition~\ref{prop:reachPSPACE-hard}]
Let us detail a polynomial reduction from the QBF problem to the contraint problem for Boolean games with Reachability objectives. Let $\psi = Q_1x_1Q_2x_2 \ldots Q_mx_m \,\phi(X)$ with $\phi(X) = C_1 \wedge \ldots \wedge C_n$ be a fully quantified Boolean formula in prenex normal form. We build the following Boolean game $\mathcal{G}_{\psi} = (\Pi,V,(V_i)_{i\in \Pi}, E, (\Gain_i)_{i\in \Pi})$ (see Figure~\ref{figure:reachPSPACEh}):

\begin{itemize}
\item	the set $V$ of vertices:
	\begin{itemize}
	\item for each variable $x_k \in X$ under quantifier $Q_k$, there exist vertices $x_k$, $\neg x_k$ and $q_k$;  
	\item for each clause $C_k$, there exist vertices $c_k$ and $t_k$;
	\item there exists an additionnal vertex $t_{n+1}$;
	\end{itemize}
\item the set $E$ of edges:
	\begin{itemize}
	\item from each vertex $q_k$ there exist an edge to $x_k$ and an edge to $\neg x_k$;
	\item from each vertex $x_k$ and $\neg x_k$, there exists an edge to $q_{k+1}$, except for $k=m$ where this edge is to $c_1$;
	\item from each vertex $c_k$, there exist an edge to $t_k$ and an edge to $c_{k+1}$, except for $k=n$ where there exist an edge to $t_n$ and an edge to $t_{n+1}$;
	\item there exists a loop on each $t_k$;
	\end{itemize}
\item the set $\Pi$ of $n+2$ players:
	\begin{itemize}
		\item each player $i$, $1 \leq i \leq  n$, owns vertex $c_i$;
		\item player $n+1$ (resp. $n+2$) is the player who owns the vertices $q_i$ for each existential (resp. universal) quantifier $Q_i$;
		\item as all other vertices have only one outgoing edge, it does not matter which player owns them;
	\end{itemize}
\item each function $\Gain_i$ is associated with the objective of visiting the set $F_i$ defined as follows:
	\begin{itemize}
	\item for all $i$, $1 \leq i \leq n$, $F_i = \{ \ell \in V \mid \ell \text{ is a literal of clause } C_i \}  \cup \{t_i\}$;
	\item $F_{n+1} = \{t_{n+1}\}$;
	\item $F_{n+2}=\{t_1,\ldots, t_n\}$.
	\end{itemize}
\end{itemize}


\begin{figure}[h!]
\centering
\scalebox{0.9}{
\begin{tikzpicture}
\node[draw] (Q1) at (0,0){$q_1$};
\node[draw] (Q2) at (3,0){$q_2$};
\node[draw] (Q3) at (6,0){$q_3$};
\node (empty) at (7.5,0){$\ldots$};
\node[draw] (Qm) at (9,0){$q_m$};
\node[draw] (C1) at (12,0){$c_1$};
\node (empty2) at (13.25,0){$\ldots$};
\node[draw](Cn) at (14.5,0){$c_n$};
\node[draw](T0) at (16.5,0){$t_{n+1}$};

\node[draw] (x1) at (1.5,1.5){$x_1$};
\node[draw] (nx1) at (1.5,-1.5){$\neg x_1$};

\node[draw] (x2) at (4.5,1.5){$x_2$};
\node[draw] (nx2) at (4.5,-1.5){$\neg x_2$};

\node[draw] (xm) at (10.5, 1.5){$x_m$};
\node[draw] (nxm) at (10.5, -1.5){$\neg x_m$};

\node[draw] (T1) at (12,2){$t_1$};
\node[draw] (Tn) at (14.5,2){$t_n$};

\draw[->] (Q1) to (x1);
\draw[->] (Q1) to (nx1);

\draw[->] (x1) to (Q2);
\draw[->] (nx1) to (Q2);

\draw[->] (Q2) to (x2);
\draw[->] (Q2) to (nx2);

\draw[->] (x2) to (Q3);
\draw[->] (nx2) to (Q3);

\draw[->](Qm) to (xm);
\draw[->](Qm) to (nxm);

\draw[->](xm) to (C1);
\draw[->](nxm) to (C1);

\draw[->](C1) to (empty2);
\draw[->](empty2) to (Cn);

\draw[->](Cn) to (T0);

\draw[->](C1) to (T1);
\draw[->](Cn) to (Tn);

\draw[->] (T1) edge [loop above] (T1);
\draw[->] (Tn) edge [loop above] (Tn);
\draw[->] (T0) edge [loop right] (T0);
\end{tikzpicture}}
\caption{Reduction from the formula $\psi$ to the Boolean game $\mathcal{G}_{\psi}$}
\label{figure:reachPSPACEh}
\end{figure}
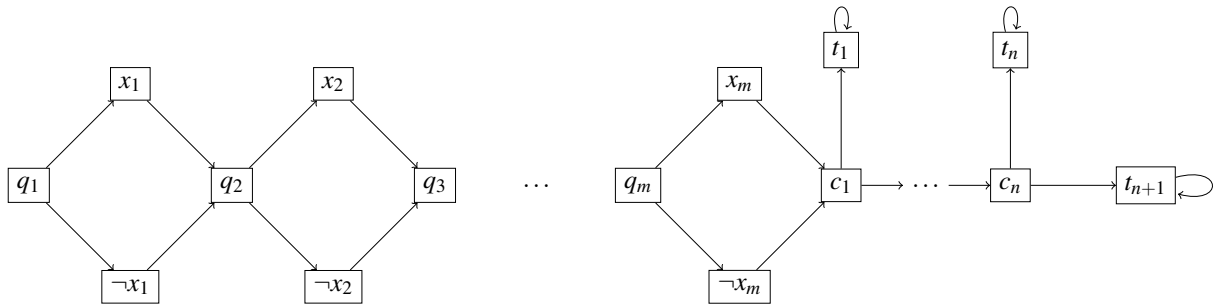

\begin{remark}
	\label{rem:winningStrat}
(1) Notice that a sequence $f$ of functions $f_k: \{0,1\}^{k-1} \rightarrow \{0,1\}$, with $k$ odd, $1 \leq k \leq m$, as presented above, can be translated into a strategy $\sigma_{n+1}$ of player~$n+1$ in the initialized game $(\mathcal{G}_{\psi},q_1)$, and conversely. Similarly, a sequence $g$ of functions $g_k: \{0,1\}^{k-1} \rightarrow \{0,1\}$, with $k$ even, $1 \leq k \leq m$ is nothing else than a strategy $\sigma_{n+2}$ of player~$n+2$. (2) Notice also that if $\rho$ is a play in $(\mathcal{G}_{\psi},q_1)$, then $\Gain_{n+1}(\rho) = 1$ if and only if $\Gain_{n+2}(\rho) = 0$. Moreover, suppose that $\rho$ visits $t_{n+1}$, then for all $i$, $1 \leq i \leq n$, $\Gain_{i}(\rho) = 1$ if and only if for all $i$, $1 \leq i \leq n$, $\rho$ visits a vertex that is a litteral of $C_i$ if and only if there is a valuation of all variables of $X$ that evaluates $\phi(X)$ to true.
\end{remark}

The game $\mathcal{G}_{\psi}$ can be constructed from $\psi$ in polynomial time. Let us now show that $\psi$ is true if and only if there exists a weak SPE in $(\mathcal{G}_{\psi},q_1)$ with a payoff $p \geq (0, \ldots, 0, 1, 0)$ (that is, such that $p_{n+1} = 1$). 

\medskip 

 $(\Rightarrow)$ Suppose that $\psi$ is true. Then there exists a sequence $f$ of functions $f_k: \{0,1\}^{k-1} \rightarrow \{0,1\}$, with $k$ odd, $1 \leq k \leq m$, such that for all sequences $g$ of functions $g_k: \{0,1\}^{k-1} \rightarrow \{0,1\}$, with $k$ even, $1 \leq k \leq m$, the valuation $\nu_{f,g}$ evaluates $\phi(X)$ to true. We define a strategy profile $\sigma$ as follows:
	\begin{itemize}
		\item for player~$n+1$, his strategy $\sigma_{n+1}$ is the strategy corresponding to the sequence $f$ (by Remark~\ref{rem:winningStrat});
		\item for player~$n+2$, his strategy is an arbitrary strategy $\sigma_{n+2}$; we denote by $g$ the corresponding sequence $g_k: \{0,1\}^{k-1} \rightarrow \{0,1\}$, with $k$ even, $1 \leq k \leq m$ 	(by Remark~\ref{rem:winningStrat});
		\item for each player~$i$, $1 \leq i \leq n$, 
		\begin{itemize}
			\item if $hv \in \Hist_{i}(q_1)$ with $v = c_i$, is consistent with $\sigma_{n+1}$, then $\sigma_i(hv)= c_{i+1}$ if $i \neq n$ and $t_{n+1}$ otherwise
			\item else $\sigma_i(hv) = t_i$.
		\end{itemize}
	\end{itemize}
	
Let us prove that $\sigma$ is a weak SPE, that is, for each history $hv \in \Hist(q_1)$, there is no one-shot deviating strategy in the subgame $(\mathcal G_{\psi\restriction h},v)$ that is profitable to the player who owns vertex $v$ (by Proposition~\ref{prop:equivWSPEetVWSPE}). This is clearly true for all $v = t_i$, $1 \leq i \leq n+1$, since $t_i$ has only one outgoing edge. For the other vertices $v$, we study two cases:

\begin{itemize}

	\item $hv$ is consistent with $\sigma_{n+1}$: By hypothesis, the valuation $\nu_{f,g}$ evaluates $\phi(X)$ to true, that is, it evaluates all clauses $C_i$ to true. Hence by Remark~\ref{rem:winningStrat}, the play $\rho = h \outcome{\rest{\sigma}{h}}{v}$ visits a vertex of $F_i$ for all $i$, $1 \leq i \leq n$, and by definition of $\sigma$, $\rho$ eventually loops on $t_{n+1}$. This means that $\Gain(\rho) = (1,1,\ldots,1,0)$. Notice that as $\sigma_{n+2}$ is arbitrary, using another strategy $\sigma'_{n+2}$ in place of $\sigma_{n+2}$ leads to a play $\rho'$ such that $\Gain(\rho') = \Gain(\rho)$. As $\Gain(\rho) = (1,1,\ldots,1,0)$, only player~$n+2$ has an incentive to deviate to increase his gain. Nevertheless, as just explained, using another strategy does not change his gain. 

	\item $hv$ is not consistent with $\sigma_{n+1}$: Suppose that $v = c_k$. Then by definition of $\sigma$, the play $h \outcome{\rest{\sigma}{h}}{v}$ eventually loops on $t_k$ leading to a gain of $1$ for player~$k$. This player has thus no incentive to deviate with a one-shot deviation in the subgame $(\mathcal G_{\psi\restriction h},v)$.  

Suppose that $v = q_k$. Then by definition of $\sigma$, the play $\rho = h \outcome{\rest{\sigma}{h}}{v}$ eventually loops on $t_1$. It follows that $\Gain_{n+1}(\rho) = 0$ and $\Gain_{n+2}(\rho) = 1$. As we only have to consider one-shot deviating strategies, if $q_k \in V_{n+2}$, player~$n+2$ has no incentive to deviate, and if $q_k \in V_{n+1}$, player~$n+1$ could try to use a one-shot deviating strategy, however the resulting  play still eventually loops on $t_1$.
	\end{itemize}	

This proves that $\sigma$ is a weak SPE. Its payoff is equal to $p = (1,1,\ldots,1,0)$ as explained previously. Therefore it satisfies the constraint $p \geq (0, \ldots, 0, 1, 0)$.

\medskip

$(\Leftarrow)$ Suppose that there exists a weak SPE $\sigma$ in $(\mathcal{G}_{\psi},q_1)$ with outcome $\rho$ and payoff $\Gain(\rho) \geq (0, \ldots, 0, 1, 0)$, that is, $\Gain_{n+1}(\rho) = 1$. By Remark~\ref{rem:winningStrat}, it follows that $\Gain_{n+2}(\rho) = 0$. We have to prove that $\psi$ is true. To this end, consider the sequence $f$ of functions  $f_k: \{0,1\}^{k-1} \rightarrow \{0,1\}$, with $k$ odd, $1 \leq k \leq m$, that corresponds to strategy $\sigma_{n+1}$ of player~$n+1$ by Remark~\ref{rem:winningStrat}. Let us show that for all sequences $g$ of functions $g_k: \{0,1\}^{k-1} \rightarrow \{0,1\}$, with $k$ even, $1 \leq k \leq m$, the valuation $\nu_{f,g}$ evaluates $\phi(X)$ to true.

By contradiction assume that it is not the case for some sequence $g'$ and consider the related strategy $\sigma'_{n+2}$ of player~$n+2$ by Remark~\ref{rem:winningStrat}. Notice that $\sigma'_{n+2}$  is a finitely deviating strategy. Let us consider the outcome $\rho'$ of the strategy profile $(\sigma'_{n+2},\sigma_{-(n+2)})$ from $q_1$. As $\Gain_{n+2}(\rho) = 0$, we must have $\Gain_{n+2}(\rho') = 0$, otherwise $\sigma'_{n+2}$ is a profitable deviation for player~$n+2$ whereas $\sigma$ is a weak SPE. It follows that $\Gain_{n+1}(\rho') = 1$ by Remark~\ref{rem:winningStrat}, that is, $\rho'$ eventually loops on $t_{n+1}$. 

Now recall that the valuation $\nu_{f,g'}$ evaluates $\phi(X)$ to false, which means that it evaluates some clause $C_k$ of $\phi(X)$ to false. Consider the history $hc_k < \rho'$. As strategy $\sigma'_{n+2}$ only acts on the left part of the underlying graph of  $\mathcal G_{\psi}$, we have $\rho' = \outcome{\sigma'_{n+2},\sigma_{-(n+2)}}{q_1} = h \outcome{\rest{\sigma}{h}}{c_k}$. In the subgame $(\mathcal{G}_{\psi\restriction h},c_k)$, the outcome of $\rest{\sigma}{h}$ gives a gain of $0$ to player~$k$ because $\rho' = h \outcome{\rest{\sigma}{h}}{c_k}$ does not visit $t_k$ and $\nu_{f,g'}$ evaluates $C_k$ to false. In this subgame, player~$k$ has thus a profitable one-shot deviation that consists to move to $t_k$. It follows that $\sigma$ is not a weak SPE which is impossible. Then $\psi$ is true.
\end{proof}

%

For Boolean games with Safety objectives, we can use the same reduction and the same kind of arguments as for Reachability objectives.

\begin{prop}
\label{prop:pspacehSafety}
	The constraint problem for Boolean games with Safety objectives is PSPACE-hard. 
\end{prop}

\begin{proof}
Given a fully quantified Boolean formula $\psi$, we construct the same game as in the proof of Proposition~\ref{prop:reachPSPACE-hard} (see Figure~\ref{figure:reachPSPACEh}), except that each player~$i$, $1 \leq i \leq n+2$, aims at avoiding the set $F'_i$ (instead of visiting the set $F_i$) defined as follows:
\begin{itemize}
	\item for all $i$, $1 \leq i \leq n$, $F'_i = \{ \ell \in V \mid \ell \text{ is a literal of clause } C_i \}  \cup \{t_{n+1}\}$;
	\item $F'_{n+1} = \{t_1, t_2, \ldots t_n\}$;
	\item $F'_{n+2} = \{t_{n+1}\}$.
\end{itemize}
Recall how the sets $F_i$ were defined: $F_{n+1} = \{t_{n+1}\}$, $F_{n+2} = \{t_1, t_2, \ldots t_n\}$, and for all $i$, $1 \leq i \leq n$, $F_i = \{ \ell \in V \mid \ell \text{ is a literal of clause } C_i \}  \cup \{t_{i}\}$. Hence we have a clear duality for players~$n+1$ and~$n+2$: a play $\rho$ visits $F_{n+1}$ (resp. $F_{n+2}$) if and only if $\rho$ avoids $F'_{n+1}$ (resp. $F'_{n+2}$). This is not the case for the others players, but one can check that the proof works in the same way as for Boolean games.   
\end{proof}

\subsubsection{Proof of Corollaries~\ref{cor:finitememory} and~\ref{cor:PSPACE}}
\label{sec:cor}

We conclude Section~\ref{sec:classes} with the proof of two previous corollaries. We begin with Corollary ~\ref{cor:PSPACE} stating that the constraint problem for SPEs (instead of weak SPEs) in Boolean games with Reachability and Safety objectives is PSPACE-complete.

\begin{proof}[Proof of Corollary ~\ref{cor:PSPACE}]
As weak SPEs and SPEs are equivalent notions for Reachability objectives (Proposition~\ref{prop:reach}), by Theorem~\ref{thm:PSPACE}, the constraint problem for SPEs for Boolean games with Reachability objectives is PSPACE-complete. 

Let us study the case of Safety objectives. The reduction from QBF proposed in the proof of Proposition~\ref{prop:pspacehSafety} uses the game $\mathcal{G}_{\psi}$ of Figure~\ref{figure:reachPSPACEh}. Due to the structure of the underlying graph, all weak SPEs of $\mathcal{G}_{\psi}$ are SPEs since any deviating strategy from a given strategy is necessarily finitely deviating. This shows that the constraint problem for SPEs is PSPACE-hard for Safety objectives. It is proved in~\cite{Ummels05} that this problem is in PSPACE. 
\end{proof}

We now turn to Corollary~\ref{cor:finitememory} that states that if there exists a weak SPE in a Boolean game, then there exists a finite-memory weak SPE with the same payoff such that the strategy memory sizes are polynomial for all objectives, except for Reachability and Safety objectives where the sizes are exponential.

\begin{proof}[Proof of Corollary ~\ref{cor:finitememory}]
For the objectives that are prefix-independent, this is an immediate consequence of Proposition~\ref{prop:lasseFolkThm} with strategy sizes in $\mathcal{O}(|V|^3 \cdot |\Pi|)$. For Reachability and Safety objectives, we need to transform the Boolean game with Reachability (resp. Safety) objectives into a Boolean game with B\"uchi (resp. Co-B\"uchi) objectives, as done in the proof of Proposition~\ref{prop:PSPACE} (resp. Proposition~\ref{prop:PSPACEsafety}). Recall that the set of vertices of the latter game is equal to $V' = V \times 2^{|\Pi|}$. We can thus again apply Proposition~\ref{prop:lasseFolkThm} and obtain strategy sizes in $\mathcal{O}(|V'|^3 \cdot |\Pi|) = \mathcal{O}(|V|^3 \cdot |\Pi| \cdot 2^{3 \cdot |\Pi|})$.
\end{proof}

\section{Fixed parameter tractability}
\label{sec:FPT}

In this section, we show that the constraint problem is $P$-complete for Explicit Muller objectives, that it is fixed parameter tractable for the other classical $\omega$-regular objectives, and that it becomes polynomial when the number of players is fixed. 

These results do not rely on the concept of good symbolic witness (as in Section~\ref{sec:classes}) but rather on the following algorithm based on Theorem~\ref{folkThm} to solve the constraint problem. Given a Boolean game $(\mathcal{G},v_0)$ and thresholds $x, y \in \{0,1\}^{|\Pi|}$, 
\begin{itemize}
\item Compute the initial sets $\rP_0(v)$, $v \in V$, and repeat the Remove-Adjust procedure (see Definition~\ref{def:remove}) until reaching the fixpoint $\rP_{k^*}(v)$, $v \in V$, 
\item Then check whether $\rP_{k^*}(v) \neq \emptyset$ for all $v \in \Succ^*(v_0)$ and whether there exists a payoff $p \in \rP_{k^*}(v_0)$ such that $x \leq p \leq y$.
\end{itemize}
We call this algorithm the \emph{decision algorithm} and its first part computing the fixpoint the \emph{fixpoint algorithm}.

\subsection{Complexity of the decision algorithm}
\label{section:complexityFixpoint}

We here study the time complexity of the decision algorithm. We express it in terms of three parameters:

\begin{itemize}
\item $\mathcal{O}(\init)$: the complexity of computing $\rP_0(v)$ for some given vertex $v$, 
\item $m = \Max_{v\in V} |\rP_0(v)|$: the maximum number of payoffs in the sets $\rP_0(v)$, $v \in V$,
\item $\mathcal{O}(path)$: the complexity of determining whether there exists a play with a given payoff $p$ from a given vertex $v$. (This test is required in both the computation of $\rP_0(v)$ and the Adjust operation.)
\end{itemize}

\begin{lemma} \label{lem:compl}
The time complexity of the decision algorithm is in $\mathcal{O}(m^3 \cdot |V| \cdot |\Pi| \cdot init \cdot path \cdot (|V| + |E|))$.
\end{lemma}

Expressing the complexity in this way will be useful in Section~\ref{sec:ExplicitMuller} (dedicated to Boolean games with Explicit Muller objectives and to the fixed parameter tractability of the constraint problem). We do not claim that the given complexity is the tightest one but this is enough for our purpose. 

The gain function $\Gain_i$ takes its values in $\{0,1\}$ for all $i$, hence $m$ is bounded by $2^{|\Pi|}$. Moreover by definition of $\rP_0(v)$, $\mathcal{O}(\init)$ is in $\mathcal{O}(2^{|\Pi|} \cdot path)$. For Boolean games with Explicit Muller objectives, we will provide a polynomial bound for both $m$ and  $\mathcal{O}(\init)$ (see Section~\ref{sec:ExplicitMuller}). Next Section~\ref{sec:path} is devoted to the study of $\mathcal{O}(path)$.

\begin{proof}[Proof of Lemma~\ref{lem:compl}] 
We suppose that the graph $G$ is given as $|V|$ lists of successors and the sets $\rP_k(v)$ are given as $|V|$ lists of payoffs. The comparison between two payoffs is in $\mathcal{O}(|\Pi|)$ time.

We first study the time complexity of the fixpoint algorithm.
\begin{itemize}
\item The first step of the algorithm is to compute the sets $\rP_0(v)$ for all $v$, that takes $\mathcal{O}(|V| \cdot init)$ time.

\item Then the algorithm repeats the Remove-Adjust procedure until reaching a fixpoint. There are at most $m \cdot |V|$ repetitions of this procedure since it removes at least one payoff from some $\rP_k(v)$.

\item The computation of one Remove operation is in $\mathcal{O}((|V| + |E|)\cdot m^2 \cdot |\Pi|)$ time. Indeed we potentially have to consider all the vertices $v$ and their successors $v'$ to check whether there exists $p \in \rP_k(v)$ such that $p_i < p'_i$ for all $p' \in \rP_k(v')$ (with $v \in V_i$). This takes $\mathcal{O}((|V| + |E|)\cdot m^2 \cdot |\Pi|)$ time. If the check is positive, we have to remove $p$ from $\rP_k(v)$ that takes $\mathcal{O}(m \cdot |\Pi|)$ time.

\item The computation of one Adjust operation is in $\mathcal{O}((|V| + |E|)\cdot m \cdot |\Pi| \cdot path)$ time. Indeed if $p$ is the payoff removed by the Remove operation just before, for all $u$ such that $p \in \rP_k(u)$, we have to check whether there exists a $(p,k)$-labeled play with payoff $p$ from $u$ and remove $p$ from $\rP_k(u)$ in case of non existence of such a play. This can be done by computing a graph $G'$ from $G$ by restricting $V$ to the vertices $u$ such that $p \in \rP_k(u)$ (in $\mathcal{O}((|V|\cdot m \cdot |\Pi| + (|V|+|E|))$ time), and for each $u$ by checking in $G'$ whether there exists a play with payoff $p$ from $u$ and then remove $p$ from $\rP_k(u)$ in case of non existence (in $\mathcal{O}(|V| \cdot (path + m \cdot |\Pi|))$ time). 
\end{itemize}

Therefore the total time complexity of the fixpoint algorithm is in $\mathcal{O}(|V| \cdot init + m \cdot |V| \cdot [(|V| + |E|)\cdot m^2 \cdot |\Pi| + (|V| + |E|)\cdot m \cdot |\Pi|\cdot path ]) = \mathcal{O}(|V| \cdot init + m^3 \cdot |V| \cdot |\Pi| \cdot path \cdot (|V| + |E|))$. This is bounded by $\mathcal{O}(m^3 \cdot |V| \cdot |\Pi| \cdot init \cdot path \cdot (|V| + |E|))$.

To get the time complexity of the decision algorithm, it remains to add the time complexities to test if $\rP_{k^*}(v) \neq \emptyset$ for all $v \in \Succ^*(v_0)$ and if there exists a payoff $p \in \rP_{k^*}(v_0)$ such that $x \leq p \leq y$. The first part is done in $\mathcal{O}((|V|+|E|))$ time by a depth-first search of $G$ from $v_0$, and the second part is done in $\mathcal{O}(m \cdot |\Pi|)$ time. The overall complexity of the decision algorithm is thus in $\mathcal{O}(m^3 \cdot |V| \cdot |\Pi| \cdot init \cdot path \cdot (|V| + |E|))$ time as announced.
\end{proof}

\subsection{Existence of a play with a given payoff}
\label{sec:path}

The purpose of this section is to prove the next lemma stating the complexity $\mathcal{O}(path)$ for all kinds of $\omega$-regular objectives.  

\begin{lemma} \label{lem:FPT}  
Let $\mathcal{G}$ be a Boolean game. Let $p \in \{0,1\}^{|\Pi|}$ and $v \in V$.
\begin{enumerate}
\item Determining whether there exists a play with payoff $p$ from $v$ is
\begin{itemize}
\item in polynomial time for B\"uchi, co-B\"uchi, Explicit Muller, and Parity objectives,
\item in $\mathcal{O}(2^{|\Pi|} (|V| + |E|))$ time for Reachability and Safety objectives, and
\item in ${\cal O}(2^L \cdot M + (L^L \cdot |V|)^5 )$ time for Rabin, Streett, and Muller objectives, where $L = 2^\ell$ and
\begin{itemize}
\item $\ell = \Sigma_{i = 1}^{|\Pi|} 2 \cdot k_i$ and $M = \mathcal{O}(\Sigma_{i = 1}^{|\Pi|} 2 \cdot k_i)$ such that for each player~$i \in \Pi$, $k_i$ is the number of his pairs $(G^i_j,R^i_j)_{1 \leq j \leq k_i}$ in the case of Rabin and Streett objectives, and
\item $\ell = \Sigma_{i = 1}^{|\Pi|} d_i$ and $M = \mathcal{O}(\Sigma_{i = 1}^{|\Pi|} |\mathcal{F}_i| \cdot d_i)$ such that for each player~$i \in \Pi$, $d_i$ (resp. $|\mathcal{F}_i|$), is the number of his colors (the size of his family of subsets of colors) in the case of Muller objectives.
\end{itemize}
\end{itemize}
\item When the number $|\Pi|$ of players is fixed, for all these kinds of objectives, the existence of a play with payoff $p$ from $v$ can be solved in polynomial time.
\end{enumerate}
\end{lemma}

The general approach to prove this lemma is the following one. A play with payoff $p$ from $v$ in a Boolean game $\mathcal{G}$ is a play satisfying an objective $\W$ equal to the conjunction of objectives $\Win_i$ (when $p_i = 1$) and of objectives $V^\omega \setminus \Win_i$ (when $p_i = 0$). It is nothing else than an infinite path in the underlying graph $G = (V,E)$ satisfying some particular $\omega$-regular objective $\W$. The existence of such paths is a well studied problem; we gather in the next proposition the known results that we need for proving Lemma~\ref{lem:FPT}. Recall that a \emph{Generalized Reachability} (resp. \emph{Generalized B\"uchi}) objective $\W$ is a conjunction of Reachability (resp. B\"uchi) objectives. Moreover, an objective $\W$ equal to a Boolean combination of B\"uchi objectives, called a \emph{BC B\"uchi objective}, is defined as follows. Let $F_1, \ldots, F_\ell$ be $\ell$ subsets of $V$, and $\phi$ be a Boolean formula over variables $f_1, \ldots, f_\ell$. We say that an infinite path $\rho$ in $G$ satisfies $(\phi,F_1, \ldots, F_\ell)$ if the truth assignment 

\centerline{$f_i = 1$ if and only if $\Inf(\rho) \cap F_i \neq \emptyset$, and $f_i = 0$ otherwise}

\noindent
satisfies $\phi$. All operators $\vee$, $\wedge$, $\neg$ are allowed in a BC B\"uchi objective. However we denote by $|\phi|$ the size of $\phi$ equal to the number of disjunctions and conjunctions inside $\phi$, and we say that the BC B\"uchi objective $(\phi,F_1, \ldots, F_\ell)$ is of size $|\phi|$ and with $\ell$ variables.

\begin{prop}
\label{prop:path}
Let $G = (V,E)$ be a graph, $v \in V$ be one of its vertices, and $\W \subseteq V^\omega$ be an objective. Then deciding the existence of an infinite path from $v$ in $G$ that satisfies $\W$ is
\begin{itemize}
\item in polynomial time for $\W$ equal to either a Streett objective $\W$, or an Explicit Muller objective, or the opposite of an Explicit Muller objective, or a conjunction of a Generalized B\"uchi objective and a co-B\"uchi objective,
\item in $\mathcal{O}(2^\ell (|V| + |E|))$ time for $\W$ equal to a conjunction of a Generalized Reachability objective and a Safety objective, where $\ell$ is the number of reachability objectives, 
\item in ${\cal O}(2^L \cdot |\phi| + (L^L \cdot |V|)^5 )$ time for a BC B\"uchi objective $\W = (\phi,F_1, \ldots, F_\ell)$, where $L = 2^\ell$.

\end{itemize}
\end{prop}

\begin{proof}
Let $\W$ be an objective. If it is a Streett objective, then the result is proved in~\cite{EmersonL87}. 

For the other objectives, we use known results about \emph{two-player zero-sum games}, where player~$A$ aims at satisfying a certain objective $\W$ whereas player $B$ tries to prevent him to satisfy it.  A classical problem is to decide whether player~$A$ has a winning strategy that allows him to satisfy $\W$ against any strategy of player~$B$, see for instance~\cite{Bruyere17,2001automata}. When player~$A$ is the only one to play, the existence of a winning strategy for him is equivalent to the existence of a path satisfying $\W$ (see \cite[Section 3.1]{Bruyere17}). This is exactly the problem that we want to solve. In the sequel of the proof, we mean by $(G,\W)$ a two-player zero-sum game, where player~$A$ (resp. player~$B$) aims at satisfying $\W$ (resp. $V^\omega \setminus \W$).

If $\W$ is an Explicit Muller objective, then deciding the existence of a winning strategy for player~$A$ (resp. player~$B$) $(G,\W)$ can be done in polynomial time by~\cite{Horn08}. Thus the case where $\W$ is the opposite of an Explicit Muller objective is also proved (by exchanging players~$A$ and $B$).

Suppose that $\W$ is the conjunction of a Generalized B\"uchi objective and a co-B\"uchi objective. By a classical reduction, the game $(G,\W)$ can be polynomially transformed into a game $(G',\W')$ with an objective $\W'$ equal to the conjunction of a B\"uchi objective and a co-B\"uchi objective. The existence of a winning strategy for player~$A$ in the latter game can be tested in polynomial time~\cite{AlfaroF07}. 

Suppose that $\W$ is the conjunction of a Generalized Reachability objective and a Safety objective, such that $\ell$ is the number of Reachability objectives and $F$ is the set of vertices to be avoided in the Safety objective. We first treat separately the Safety objective by removing from $G$ all the vertices of $F$. This can be done in $\mathcal{O}(|V| + |E|)$ time. In the resulting graph $G'$, we then test the existence of a winning strategy for player~$A$ in the game $(G',\W')$ with $\W'$ being the Generalized Reachability objective. This can be done in $\mathcal{O}(2^{\ell} (|V| + |E|))$ time~\cite{FijalkowH13}. 

If $\W$ is a BC B\"uchi objective $(\phi,F_1, \ldots, F_\ell)$, then deciding the existence of a winning strategy for player~$A$ in the game $(G,\W)$ can be done in ${\cal O}(2^L \cdot |\phi| + (L^L \cdot |V|)^5 )$ time with $L = 2^\ell$ by~\cite{BruyereHR17}. 
\end{proof}

\begin{proof}[Proof of Lemma~\ref{lem:FPT}]
Let us prove Part 1 of the lemma. A play with payoff $p$ from $v$ in $\mathcal{G}$ is a play satisfying an objective $\W$ equal to the conjunction of objectives $\Win_i$ (when $p_i = 1$) and of objectives $V^\omega \setminus \Win_i$ (when $p_i = 0$). For each type of Boolean objectives $ \Win_i$, we first explain what kind of objective $\W$ we obtain and we then apply Proposition~\ref{prop:path}. 
\begin{itemize}

\item Consider a Boolean game $\mathcal{G}$ with Parity objectives. In this case, as $\Win_i$ is a Parity objective for all $i \in \Pi$, each $V^\omega \setminus \Win_i$ is again a Parity objective, and $\W$ is thus a conjunction of Parity objectives which is a Streett objective~\cite{ChatterjeeHP07}. Therefore the existence of a play with payoff $p$ in $\mathcal{G}$ can be tested in polynomial time by Proposition~\ref{prop:path}.

\item Consider the case of B\"uchi objectives. Then, the intersection of B\"uchi objectives $\Win_i$ (when $p_i = 1$) is a Generalized B\"uchi objective and the intersection of co-B\"uchi objectives $V^\omega \setminus \Win_i$ (when $p_i = 0$) is again a co-B\"uchi objective. Hence $\W$ is the conjunction of a Generalized B\"uchi objective and a co-B\"uchi objective. The existence of a play with payoff $p$ in $\mathcal{G}$ can be tested in polynomial time~by Proposition~\ref{prop:path}. 

Notice that the case of Boolean games with co-B\"uchi objectives is solved exactly in the same way. Indeed we have the same kind of objective $\W$ since $\Win_i$ is a co-B\"uchi objective if and only if $V^\omega \setminus \Win_i$ is a B\"uchi objective. 

\item Consider a Boolean game with Reachability objectives. The intersection of Reachability objectives $\Win_i$ (when $p_i = 1$) is a Generalized Reachability objective and the intersection of Safety objectives $V^\omega \setminus \Win_i$ (when $p_i = 0$) is again a Safety objective. The existence of a play with payoff $p$ in $\mathcal{G}$ can be tested in $\mathcal{O}(2^{|\Pi|} (|V| + |E|))$ time by Proposition~\ref{prop:path} as there are at most $|\Pi|$ Reachability objectives. 

The case of Boolean games with Safety objectives is solved in the same way. 

\item Consider a Boolean game with Rabin objectives (with the related families $(G^i_j,R^i_j)_{1 \leq j \leq k_i}$, $i \in \Pi$). In this case, the objective $\W$ is the conjunction of Rabin objectives (when $p_i = 1$) and of Streett objectives (when $p_i = 0$), that is, $\W$ is a BC B\"uchi objective $(\phi,(G^i_j,R^i_j)_{1 \leq j \leq k_i}, i \in \Pi)$ such that
\begin{eqnarray} \label{eq:Rabin}
\phi &=& \bigwedge_{i \mid p_i = 1} \bigvee_{j = 1}^{k_i} ( g^i_j \wedge \neg r^i_j)  ~~\wedge~~ \bigwedge_{i \mid p_i = 0} \bigwedge_{j = 1}^{k_i}  ( \neg g^i_j \vee  r^i_j)
\end{eqnarray}
In this formula, the variable $g^i_j$ (resp. $r^i_j$) is associated with the set $G^i_j$ (resp. $R^i_j$), and $\phi$ has size $\mathcal{O}(\Sigma_{i = 1}^{|\Pi|} 2 \cdot k_i)$ and has $\Sigma_{i = 1}^{|\Pi|} 2 \cdot k_i$ variables. The announced complexity for deciding the existence of a play with payoff $p$ follows from Proposition~\ref{prop:path}.

The case of Boolean games with Streett objectives is solved in the same way.

\item The case of Boolean games with Muller objectives (with the related coloring functions $\Omega_i : V \rightarrow \{1,\ldots,d_i\}$ and families $\mathcal{F}_i \subseteq 2^{\Omega_i(V)}$, $i \in \Pi$) is treated as in the previous item. Indeed a play satisfies the Muller objective $\Win_i$ if there exists an element $F$ of $\mathcal{F}_i$ such that all colors of $F$ are seen infinitely often along the play while no other color is seen infinitely often. Therefore, as the objective $\W$ is a conjunction of Muller objectives and of the opposite of Muller objectives, $\W$ is a BC B\"uchi objective $(\phi,(F^i_c)_{c \in \{1,\ldots,d_i\}, i \in \Pi})$ described by the following formula $\phi$ 
\begin{eqnarray} \label{eq:Muller}
\phi &=& \bigwedge_{i \mid p_i = 1} \bigvee_{F \in \mathcal{F}_i} ( \bigwedge_{c \in F} f^i_c ~\wedge \bigwedge_{c \not\in F} \neg f^i_c)  
~~\wedge~~ \bigwedge_{i \mid p_i = 0} \bigwedge_{F \in \mathcal{F}_i} ( \bigvee_{c \in F} \neg f^i_c ~\vee \bigvee_{c \not\in F} f^i_c) 
\end{eqnarray}
In this formula, the variable $f^i_c$ is associated with the subset $F^i_{c} = \{ v \in V \mid \Omega_i(v) = c \}$ of vertices colored by color $c \in \{1,\ldots,d_i\}$, $i \in \Pi$.
This formula has size $\mathcal{O}(\Sigma_{i = 1}^{|\Pi|} |\mathcal{F}_i| \cdot d_i)$ and has $\Sigma_{i = 1}^{|\Pi|} d_i$ variables. 

\item It remains to treat the case of Boolean games with Explicit Muller objectives (with the related families $\mathcal{F}_i \subseteq 2^V$, $i \in \Pi$). The approach is a little different in a way to get a polynomial algorithm. By definition, there exists a play $\rho$ with payoff $p$ if and only if for all $i$, $F = \Inf(\rho) \in \mathcal{F}_i$ exactly when $p_i = 1$. 

If $p \neq (0, \ldots, 0)$, such potential sets $F$ can be computed as follows. Initially let $\W$ be an empty set. Then for each $F \in \cup_{i\in \Pi} \mathcal{F}_i$, we compute $q \in \{0,1\}^{|\Pi|}$ such that $q_i = 1$ if and only if $F \in \mathcal{F}_i$, and if $p = q$ we add $F$ to $\W$. Notice that $\W$ can be computed in polynomial time. Hence to test the existence of a play with payoff $p$ in $\mathcal G$, we test the existence of a path in $G$ satisfying the Explicit Muller objective $\W$. This can be done in polynomial time by Proposition~\ref{prop:path}.

If $p = (0, \ldots, 0)$, there exists a play $\rho$ with payoff $p$ if and only if no $F \in \W' = \cup_{i\in \Pi} \mathcal{F}_i$ is equal to $\Inf(\rho)$, \emph{i.e.}, if and only if there exists a path in $G$ satisfying the opposite of the Explicit Muller objective $\W'$. This can be tested in polynomial time by Proposition~\ref{prop:path}.
\end{itemize}

We now turn to Part 2 of the lemma. Suppose that the number $|\Pi|$ of players is fixed. In case of Boolean games with Reachability, Safety, B\"uchi, co-B\"uchi, Explicit Muller, and Parity objectives, by Part~$1$ of the lemma, we get a polynomial time algorithm for deciding whether there exists a play with payoff $p$. In case of Boolean games with Rabin, Streett, and Muller objectives, we need another argument in view of the complexities of Part~$1$ of the lemma. 

Let us begin with Rabin objectives (the argument is similar for Streett objectives). Recall that we are faced with a BC B\"uchi objective described by formula $\phi$ given in (\ref{eq:Rabin}). This formula is a conjunction of disjunctions of subformulas of the form either $g^i_j \wedge \neg r^i_j$, or $\neg g^i_j$, or $r^i_j$. It can be rewritten as a disjunction of conjunctions of those subformulas: $\bigvee_{r}\bigwedge_{s}\phi_{rs}$. In the latter formula, 
\begin{itemize}
\item the number of subformulas $\bigwedge_{s}\phi_{rs}$ is polynomial since it is equal to $\Pi_{i \mid p_i = 1} k_i \cdot \Pi_{i \mid p_i = 0} 2k_i$, and 
\item each subformula $\bigwedge_{s}\phi_{rs}$ describes a polynomial conjunction of B\"uchi and co-B\"uchi objectives, that is, an objective $\W_r$ equal to the conjunction of a Generalized B\"uchi objective and a co-B\"uchi objective. 
\end{itemize}
By Proposition~\ref{prop:path}, testing whether there exists a path satisfying $\W_r$ can be done in polynomial time. Therefore testing the existence of a path satisfying the objective described by $\phi$ reduces to a polynomial number of tests (disjunction $\bigvee_{r}$) that can be done in polynomial time (objective $\W_r$).

The argument for Muller objectives is similar. The formula $\phi$ given in (\ref{eq:Muller}) is a conjunction of disjunctions of subformulas of the form either $ \bigwedge_{c \in F} f^i_c ~\wedge \bigwedge_{c \not\in F} \neg f^i_c$, or $\neg f^i_c$, or $f^i_c$. Rewriting $\phi$ as a disjunction of conjunctions of those subformulas, that is, as $\bigvee_{r}\bigwedge_{s}\phi_{rs}$, we get again a polynomial number ($= \Pi_{i \mid p_i = 1} | \mathcal{F}_i | \cdot \Pi_{i \mid p_i = 0}( |\mathcal{F}_i| \cdot d_i)$) of subformulas $\bigwedge_{s}\phi_{rs}$, each of them describing a polynomial conjunction of B\"uchi and co-B\"uchi objectives.
\end{proof}

\subsection{P-completeness for Explicit Muller objectives and FPT for the other objectives}
\label{sec:ExplicitMuller}

With the complexity of the decision algorithm given in Lemma~\ref{lem:compl} and the study of $\mathcal{O}(path)$ made in  Lemma~\ref{lem:FPT}, we are now ready to show that the constraint problem is P-complete for Explicit Muller objectives, and that for the other objectives, it is fixed parameter tractable and becomes polynomial when the number $|\Pi|$ of players is fixed.

\begin{thm}
	The constraint problem for multiplayer Boolean games with Explicit Muller objectives is P-complete.	
\end{thm}

\begin{proof}
We denote by $\mathcal{F}_i \subseteq 2^V$, the family of each player $i \in \Pi$ for his Explicit Muller objective. Let us prove the P-easyness.
By Lemma~\ref{lem:compl}, the decision algorithm is in $\mathcal{O}(m^3 \cdot |V| \cdot |\Pi| \cdot init \cdot path \cdot (|V| + |E|))$ time, where $m = \Max_{v\in V} |\rP_0(v)|$, $\mathcal{O}(\init)$ is the time complexity of computing $\rP_0(v)$ for some given vertex $v$, and $\mathcal{O}(path)$ is the time complexity of testing whether there exists a play with a given payoff $p$ from a given vertex $v$. Lemma~\ref{lem:FPT} states that $\mathcal{O}(path)$ is polynomial for Boolean games with Explicit Muller objectives. To establish the P-easiness, it remains to prove that $\mathcal{O}(m)$ and $\mathcal{O}(init)$ are also polynomial. 
First, if there exists a play $\rho$ with payoff $p$ from $v$, then either $\Inf(\rho) \in \cup_{i\in \Pi} {\mathcal{F}_i}$ or $p = (0,\ldots,0)$. Thus 
$$\rP_0(v) \subseteq P = \{q \in \{0,1\}^{|\Pi|} \mid \exists F \in \cup_{i\in \Pi} {\mathcal{F}_i}, q_i = 1 \Leftrightarrow F \in \mathcal{F}_i \} \cup \{(0,\ldots,0)\}.$$    
(This kind of argument was already used in the proof of Lemma~\ref{lem:FPT} for Explicit Muller objectives). Therefore, $|\rP_0(v)|  \leq  |P| \leq |\cup_{i\in \Pi} \mathcal{F}_i| + 1$, showing that $\mathcal{O}(m)$ is polynomial. 
Second, to compute $\rP_0(v)$, we check for each $p \in P$ whether there exists a play with payoff $p$ from $v$. It follows that $\mathcal{O}(init)$ is polynomial by Lemma~\ref{lem:FPT}.

The P-hardness is obtained thanks to a reduction from the AND-OR graph reachability problem that is P-complete~\cite{Immerman81}. Indeed, the P-hardness of the constraint problem for SPEs (instead of weak SPEs) in Boolean games with Reachability objectives is proved in~\cite[Corollary 6.22]{Ummels05} thanks to such a reduction, and it is not difficult to see that the same reduction also holds for weak SPEs and Explicit Muller objectives. 
\end{proof}


We recall that a \emph{parameterized language} $L$ is a subset of $\Sigma^* \times \mathbb{N}$, where $\Sigma$ is a finite alphabet, the second component being the parameter of the language. It is called \emph{fixed parameter tractable} if there is an algorithm that determines whether $(x,t) \in L$ in $f(t) \cdot |x|^c$ time, where $c$ is a constant independent of the parameter $t$ and $f$ is a computable function depending on $t$ only. We also say that $L$ belongs to the class FPT. Intuitively, a language is in FPT if there is an algorithm running in polynomial time with respect to the input size times some computable function on the parameter. 
We refer the interested reader to \cite{DowneyF99} for more details on parameterized complexity.

\begin{thm} \label{thm:FPT}
Let $\mathcal{G}$ be a Boolean game.
\begin{enumerate}
\item The constraint problem is in FPT for Reachability, Safety, B\"uchi, co-B\"uchi, Parity, Muller, Rabin, and Streett objectives. The parameters are
\begin{itemize}
\item the number $|\Pi|$ of players for Reachability, Safety, B\"uchi, co-B\"uchi, and Parity objectives,
\item the number $|\Pi|$ of players and the numbers $k_i$, $i \in \Pi$, of pairs $(G^i_j,R^i_j)_{1 \leq j \leq k_i}$, for Rabin and Streett objectives, and
\item the number $|\Pi|$ of players, the numbers $d_i$, $i \in \Pi$, of colors and the sizes $|\mathcal{F}_i|$, $i \in \Pi$, of the families of subsets of colors for Muller objectives.
\end{itemize}
\item When the number $|\Pi|$ of players is fixed, for all these kinds of objectives, the constraint problem can be solved in polynomial time.
\end {enumerate}
\end{thm}

Notice that in this theorem, to obtain fixed parameter tractability for Rabin, Streett, and Muller objectives, in addition to the number of players, we also have to consider the parameter equal to the size of the objective description. Nevertheless, when the number of players is fixed, we get polynomial time complexity for all types of objectives.

\begin{proof}[Proof of Theorem~\ref{thm:FPT}]
We again use Lemmas~\ref{lem:compl} and~~\ref{lem:FPT} as in the previous proof. By Lemma~\ref{lem:compl}, the complexity of the \fixpointalgo\ is in $\mathcal{O}(m^3 \cdot |V| \cdot |\Pi| \cdot init \cdot path \cdot (|V| + |E|))$ 
time. The complexity $\mathcal{O}(path)$ is studied in Part 1 (resp. in Part 2 when $|\Pi|$ is fixed) of Lemma~\ref{lem:FPT}, and we have $m \leq 2^{|\Pi|}$ and $\mathcal{O}(\init) = \mathcal{O}(2^{|\Pi|} \cdot path)$. Hence we get both parts of Theorem~\ref{thm:FPT} as a consequence of these results. 
\end{proof}

\section{Conclusion and future work}
\label{sec:conc}
In this paper, we have studied the computational complexity of the constraint problem for weak SPEs. We were able to obtain precise complexities for all the classical classes of $\omega$-regular objectives (see Table~\ref{tab:complexity}), with one exception: we have proved NP-membership for B\"uchi objectives and failed to prove NP-hardness. We have also shown that the constraint problem can be solved in polynomial time when the number of players is fixed. Finally, we have provided some fixed parameter tractable algorithms when the number of players is considered as a parameter of the problem, for Reachability, Safety, B\"uchi, Co-B\"uchi, and Parity  objectives. For the other Rabin, Streett, and Muller objectives, we also had to consider the size of the objective description as a parameter to obtain fixed parameter tractability. In a future work, we want to understand if the use of this second parameter is really necessary.

By characterizing the exact complexity of the constraint problem for Reachability and Safety objectives, we have obtained that this problem for SPEs (as for weak SPEs) is PSPACE-complete for those objectives. In the future, we intend to investigate the complexity of the other classes of $\omega$-regular objectives for SPEs. It would be also interesting to extend the study to quantitative games. For instance the constraint problem for (weak) SPEs in reachability quantitative games is decidable~\cite{BrihayeBMR15} but its complexity is unknown.

\bibliographystyle{eptcs} 

\bibliography{biblio}



	
\end{document}